\documentclass[11pt]{article}
\usepackage{times}
\usepackage{amssymb}
\usepackage{amsmath}
\usepackage{xspace,multirow,comment,bbm}
\usepackage[shortlabels,inline]{enumitem}
\usepackage{calc}
\usepackage{amsthm}
\usepackage{graphicx}
\usepackage[ruled, linesnumbered, vlined, nokwfunc]{algorithm2e}
\usepackage{caption,makecell}
\usepackage{xcolor}
\usepackage{graphpap}

\newif\ifdvi
\ifdvi
\else
\usepackage[backref=page]{hyperref}
\hypersetup{
	bookmarks=true,
	bookmarksnumbered=true,
	bookmarksopen=true
}

\fi

\makeatletter
\newcommand{\clonelabel}[2]{\@bsphack
  \expandafter\ifx\csname r@#2\endcsname\relax
  \else\protected@write\@auxout{}{\string\newlabel{#1}%
    {\csname r@#2\endcsname}}%
  \fi
  \expandafter\ifx\csname r@#2@cref\endcsname\relax
  \else\protected@write\@auxout{}{\string\newlabel{#1@cref}%
    {\csname r@#2@cref\endcsname}}%
  \fi
  \@esphack}

\def\@citex[#1]#2{\leavevmode
  \let\@citea\@empty
  \@cite{\@for\@citeb:=#2\do
    {\@citea\def\@citea{,\penalty\@m\ }%
\edef\magic##1{\let##1\expandafter\noexpand\csname bibalias@\@citeb\endcsname}%
\magic\tmp \ifx\tmp\relax\else \let\@citeb\tmp\fi
     \edef\@citeb{\expandafter\@firstofone\@citeb\@empty}%
     \if@filesw\immediate\write\@auxout{\string\citation{\@citeb}}\fi
     \@ifundefined{b@\@citeb}{\hbox{\reset@font\bfseries ?}%
       \G@refundefinedtrue
       \@latex@warning
         {Citation `\@citeb' on page \thepage \space undefined}}%
       {\@cite@ofmt{\csname b@\@citeb\endcsname}}}}{#1}}
\def\bibalias#1#2{\expandafter\def\csname bibalias@#1\endcsname{#2}}
\makeatother

\newcommand{\np}{{\em NP}\xspace}
\newcommand{\nphard}{\np-hard\xspace} 

\newcommand{\apx}{{\em APX}\xspace}
\newcommand{\apxhard}{\apx-hard\xspace}
\newcommand{\p}{{\em P}\xspace}

\DeclareMathOperator{\supp}{supp}

\DeclareMathOperator{\spn}{span}

\newenvironment{proofof}[1]{\begin{proof}[Proof of {#1}]}{\end{proof}}

\newtheorem{theorem}{Theorem}[section]
\newtheorem{lemma}[theorem]{Lemma}

\newtheorem{claim}[theorem]{Claim}

\newtheorem{corollary}[theorem]{Corollary}

{\theoremstyle{definition} 
\newtheorem{definition}[theorem]{Definition}

\newtheorem{remark}[theorem]{Remark}}

\newcommand{\R}{\ensuremath{\mathbb R}}
\newcommand{\Z}{\ensuremath{\mathbb Z}}

\newcommand{\T}{\ensuremath{\mathcal T}}
\newcommand{\Lc}{\ensuremath{\mathcal L}}

\newcommand{\Sc}{\ensuremath{\mathcal S}}

\newcommand{\OPT}{\ensuremath{\mathit{OPT}}}

\newcommand{\es}{\ensuremath{\emptyset}}
\newcommand{\assign}{\ensuremath{\leftarrow}}

\newcommand{\ceil}[1]{\ensuremath{\left\lceil#1\right\rceil}}
\newcommand{\floor}[1]{\ensuremath{\left\lfloor#1\right\rfloor}}

\newcommand{\gm}{\ensuremath{\gamma}}

\newcommand{\sse}{\subseteq}

\newcommand{\hx}{\ensuremath{\widehat x}}

\newcommand{\bx}{\ensuremath{\overline x}}

\newcommand{\bB}{\ensuremath{\overline B}}

\newcommand{\al}{\ensuremath{\alpha}}
\newcommand{\tht}{\ensuremath{\theta}}

\newcommand{\dt}{\ensuremath{\delta}}

\newcommand{\ve}{\ensuremath{\varepsilon}}

\newcommand{\sndp}{\ensuremath{\mathsf{SNDP}}\xspace}
\newcommand{\res}{\ensuremath{\mathsf{res}}\xspace}
\newcommand{\kecss}[1][k]{\ensuremath{{#1}\text{-}\mathsf{ECSS}}\xspace}
\newcommand{\kecsm}[1][k]{\ensuremath{{#1}\text{-}\mathsf{ECSM}}\xspace}
\newcommand{\kecsslp}[1][k]{\ensuremath{{#1}\text{\textnormal{-ECSSLP}}}\xspace}
\newcommand{\kecsmlp}[1][k]{\ensuremath{{#1}\text{\textnormal{-ECSMLP}}}\xspace}
\newcommand{\lpopt}[1][{\kecsslp}]{\ensuremath{\mathsf{LP}^*_{{#1}}}\xspace}
\newcommand{\mdkecss}[1][k]{\ensuremath{\mathsf{MD}\text{-}\kecss[{#1}]}\xspace}
\newcommand{\mdkecsm}[1][k]{\ensuremath{\mathsf{MD}\text{-}\kecsm[{#1}]}\xspace}
\newcommand{\mdkecsslp}[1][k]{\ensuremath{\text{\textnormal{MD-}}\kecsslp[{#1}]}\xspace}
\newcommand{\mdkecsmlp}[1][k]{\ensuremath{\text{\textnormal{MD-}}\kecsmlp[{#1}]}\xspace}

\newcommand{\ldeg}{\ell}
\newcommand{\udeg}{b}
\newcommand{\lb}{\ensuremath{\mathsf{lb}}}
\newcommand{\ub}{\ensuremath{\mathsf{ub}}}
\newcommand{\after}{\ensuremath{\mathsf{iterend}}\xspace}

\SetAlgoProcName{Algorithm}{Algorithm}
\SetFuncSty{textsc}
\SetAlCapNameSty{textsc}
\SetKwComment{simpc}{// }{}
\SetCommentSty{textnormal}
\DontPrintSemicolon

\title{Almost Tight Additive Guarantees for \boldmath $k$-Edge-Connectivity%
\thanks{To appear in the Proceedings of FOCS 2025.}} 

\author{
    Nikhil Kumar\thanks{{\tt\{nikhil.kumar2,cswamy\}@uwaterloo.ca}.
    Dept. of Combinatorics and Optimization, Univ. Waterloo, Waterloo, ON N2L 3G1. 
    Supported in part by C. Swamy's NSERC Discovery grant.}
\and 
\addtocounter{footnote}{-1} 
    Chaitanya Swamy\footnotemark
}
\date{}

\begin{document}

\maketitle

\def\thepage{}
\thispagestyle{empty}

\bibalias{LauNSS09}{lau2009survivable}
\bibalias{HershkowitzKZ24}{hershkowitz2024ghost}
\bibalias{GabowGTW09}{GGTW09}
\bibalias{GabowG12}{GG12}

\begin{abstract}
We consider the \emph{$k$-edge connected spanning subgraph} (\kecss) problem, where we are
given an undirected graph \( G = (V, E) \) with nonnegative edge costs \( \{c_e\}_{e \in
  E} \), and the goal is to find a minimum-cost subgraph \( H \) of \( G \) that is
\emph{$k$-edge connected}, i.e., there exist at least \( k \) edge-disjoint paths between
every pair of vertices in \( H \). For even \( k \), we present a polynomial time
algorithm that computes a {\em \((k-2)\)-edge connected subgraph of cost at most that of
the optimal \( k \)-edge connected subgraph of $G$}; for odd $k$, we obtain a $(k-3)$-edge 
connected subgraph of cost at most the optimum. In fact, the cost of our solution does not
exceed the optimal value, $\lpopt$ of the natural LP-relaxation for \kecss. Since
\kecss is \apxhard for all values of $k\geq 2$, our results are nearly optimal. They also 
significantly improve upon the recent work of Hershkowitz, Klein, and
Zenklusen~\cite{hershkowitz2024ghost}, both in terms of solution quality and the
simplicity of algorithm and its analysis. 
Interestingly, our techniques also yield an alternate guarantee, where we obtain a
$(k-1)$-edge connected subgraph of cost at most $1.5\cdot\lpopt$; with unit edge costs,
the cost guarantee improves to $\bigl(1+\frac{4}{3k}\bigr)\cdot\lpopt$, which improves
upon the state-of-the-art approximation guarantee for unit edge costs~\cite{GabowG12},
albeit with a unit loss in edge connectivity. 

Our \kecss-result also yields results for the \emph{$k$-edge
connected spanning multigraph} (\kecsm) problem, where multiple copies of an edge can be
selected. For \kecsm, we obtain a \(\left(1 + \frac{2}{k}\right)\)-approximation algorithm
for even $k$, and a \(\left(1 + \frac{3}{k}\right)\)-approximation algorithm for odd $k$. 

Finally, our techniques extend to the {\em degree-bounded} versions of \kecss and \kecsm,
wherein we also impose degree lower- and upper- bounds on the nodes.
Our results for \kecss and \kecsm extend to yield the same cost and connectivity
guarantees for these degree-bounded versions with an additive violation of (roughly) $2$ 
for the degree bounds.  
These are the first results for degree-bounded \mbox{$\{\kecss,\kecsm\}$} of the form where
the cost of the solution obtained is at most the optimum, and the connectivity 
constraints are violated by an additive constant. 
\end{abstract}

\newpage \pagenumbering{arabic} \normalsize

\section{Introduction} \label{intro}
We consider the {\em $k$-edge connected spanning subgraph} (\kecss) problem, wherein we
are given an undirected graph $G=(V,E)$ and nonnegative edge costs $\{c_e\}_{e\in E}$, and
we seek a minimum-cost subgraph $H$ of $G$ that is $k$-edge connected, i.e., for any
$u,v\in V$, there are $k$ edge-disjoint $u$-$v$ paths in $H$. We overload
notation and use $H$ to denote both the subgraph of $G$ and its edge-set.
A closely related problem is the {\em $k$-edge connected spanning multigraph} (\kecsm)
problem, wherein we are allowed to pick multiple copies of an edge, i.e., the subgraph $H$
is allowed to be a multigraph.%
\footnote{As in most prior work on these problems with general edge costs, we
allow $G$ itself to be a multigraph, i.e., have potentially parallel edges. But, in
\kecss, we may pick each edge at most once, whereas in \kecsm, we may pick an edge
multiple times.} 
When $k=1$, both \kecss and \kecsm correspond to the well known minimum spanning tree
problem, which can be solved exactly in polynomial time. For $k \geq 2$, both 
problems are \apxhard~\cite{Pri11}, i.e., there is a universal constant $\epsilon >0$ such
that there does not exist a $(1+ \epsilon)$-approximation algorithm, unless \p = \np.
\kecss and \kecsm are classical network-design problems with a rich history of study; see,
e.g.~\cite{FJ81,FJ82,CT00,GabowGTW09,GG12,Pri11,KKO21a,KKO22,hershkowitz2024ghost} and the
references therein.

For $k\geq 2$, the best-known approximation guarantee for \kecss is a longstanding
$2$-approximation algorithm due to Frederickson and J\'aJ\'a~\cite{FJ81,FJ82}. This 
also follows from Jain's seminal work~\cite{Jai01} that developed a $2$-approximation
algorithm for the more general survivable network-design problem and introduced the
elegant and versatile iterative-rounding technique. 
Very recently, Hershkowitz, Klein and Zenklusen~\cite{hershkowitz2024ghost} extended the
iterative-rounding framework of Jain~\cite{Jai01} and obtained the following complementary
result for \kecss: they showed that one can obtain, in polynomial time, a $(k-10)$-edge
connected subgraph $H$ of cost at most the optimum. That is, instead of blowing up the
cost compared to the optimal solution, they relax the connectivity constraints by an 
{\em additive constant} (equal to $10$).  
Note that since \kecss is \nphard, unless \p = \np, any polynomial time algorithm must
blow up the cost or violate feasibility. 
This also leads to a $\bigl(1+\frac{10}{k}\bigr)$-approximation algorithm for
\kecsm, which is asymptotically tight in the sense that~\cite{hershkowitz2024ghost} 
proved a $\bigl(1+\frac{\Omega(1)}{k}\bigr)$-factor hardness-of-approximation for \kecsm,
assuming \p$\neq\,$\np.

\subsection{Our results} \label{contrib}
In this paper, we significantly improve the above bounds obtained by Hershkowitz, Klein
and Zenklusen~\cite{hershkowitz2024ghost}, and obtain {\em a nearly-tight 
result for \kecss}. We show that for \kecss, when $k$ is even, one can obtain in
polynomial time a subgraph $H$ of 
{\em cost at most the optimum that is $(k-2)$-edge connected}. Since 
\kecss remains \nphard for even $k$, observe that if we want the cost to be at most the
optimum, then the best we can hope for is a $(k-1)$-edge-connected
subgraph, and so our result is almost tight. The result for even $k$ also
immediately implies that for odd $k$, we obtain a $(k-3)$-edge-connected subgraph of cost
at most the optimum. Our result is based on LP-rounding, and the cost of our solution is
in fact at most the optimal value $\lpopt$ of the following natural LP-relaxation 
\eqref{kecss-lp} for \kecss. 
\begin{alignat}{3}
\min & \quad & \sum_{e\in E}c_ex_e \tag{\kecsslp} \label{kecss-lp} \\
\text{s.t.} \quad && x\bigl(\dt(S)\bigr) & \geq k
\qquad && \forall \es\neq S\subsetneq V \label{con} \\
&& 0 \leq x_e & \leq 1 && \forall e\in E. \label{bnd}
\end{alignat}
\begin{theorem}[Proved in Section~\ref{kecss}] \label{mainthm} \label{kecss-thm}
Algorithm~\ref{kecss-alg} in Section~\ref{kecss} runs in polynomial time and, for any even
$k$, returns a $(k-2)$-edge connected subgraph $H$ that satisfies
$c(H)\leq\lpopt$. 

For odd $k$, running Algorithm~\ref{kecss-alg} with $k-1$, we obtain a $(k-3)$-edge
connected subgraph $H$ with $c(H)\leq\lpopt$.
\end{theorem}

We thus improve upon the guarantees obtained in~\cite{HershkowitzKZ24}; notably, our 
improvements are obtained via a much simpler algorithm and analysis. We show that the
approach leading to the above guarantee can be modified to obtain the following
interesting alternate tradeoff between connectivity and cost.

\begin{theorem}[Proved in Section~\ref{bicriteria}] \label{bicriteria-kecss}
There is a polynomial-time algorithm that returns a $(k-1)$-edge connected subgraph $H$
with $c(H)\leq 1.5\cdot\lpopt$. For unweighted \kecss, i.e., when $c_e=1$ for all 
$e\in E$, the cost bound improves to 
$c(H)\leq\min\bigl\{1.5,\bigl(1+\frac{4}{3k}\bigr)\bigr\}\cdot\lpopt$.
\end{theorem}

For unweighted \kecss, the approximation factor above improves upon the state-of-the-art
approximation factor of $\bigl(1+\frac{1.91}{k}\bigr)$ for unweighted \kecss due to Gabow
and Gallagher~\cite{GabowG12}, albeit by sacrificing ``one unit'' of edge-connectivity.
Our results thus paint a nuanced picture of the approximation vs. connectivity tradeoff
for \kecss, as depicted in the following figure.

\begin{figure}[ht!]
\centering
\hspace*{-0.25in}
\resizebox{!}{3in}{\input{tradeoff-kecss-arxiv.pdf_t}}
\captionsetup{font=small, labelfont=small}
\caption{The approximation vs. connectivity tradeoff for (a) \kecss, and (b) unweighted
  \kecss. The shaded circles in (a) and the shaded squares in (b) denote our results, as
  stated in Theorems~\ref{kecss-thm} and~\ref{bicriteria-kecss}.  
  The ? denotes the point $(k-1,1)$, which is the tightest polytime-achievable
  guarantee of the form $(\cdot,1)$. The results in~\cite{NutovC25} were obtained
  independently and concurrently; neither group was aware of the others' work
  when they obtained their respective results.
  For unweighted \kecss: Gabow and Gallagher~\cite{GabowG12} actually obtain a
  $\bigl(1+\frac{1.91}{k}\bigr)$-approximation; the results in~\cite{GabowGTW09,GabowG12}
  are not stated as yielding the $(k-2,1)$-guarantee, but they can be translated to obtain
  this guarantee; see Section~\ref{relwork}. \label{tradeoff}}   
\end{figure}

The \kecss-result in Theorem~\ref{kecss-thm} also implies a
$\bigl(1+\frac{3}{k}\bigr)$-approximation guarantee  
for \kecsm, improving upon the $\bigl(1+\frac{10}{k}\bigr)$-approximation obtained
by Hershkowitz et al.~\cite{hershkowitz2024ghost}. The natural LP-relaxation for \kecsm
replaces constraints \eqref{bnd} with the constraints $x\geq 0$. Let $\lpopt[{\kecsmlp}]$
denote the optimal value of this LP.

\begin{theorem}[Proved in Section~\ref{kecsm}] \label{kecsm-thm}  
There is a polynomial time algorithm that 
returns a feasible \kecsm solution $H$ such that
$c(H)\leq\rho_k\cdot\lpopt[{\kecsmlp}]$,
where $\rho_k=\bigl(1+\frac{2}{k}\bigr)$, if $k$ is even, and
$\bigl(1+\frac{3}{k}\bigr)$ if $k$ is odd.
\end{theorem}

Theorem~\ref{kecsm-thm} follows from Theorem~\ref{kecss-thm} because if
$x^*$ is an optimal solution to the LP-relaxation for \kecsm, then, say for even $k$, one
can consider the solution $x=\frac{k+2}{k}\cdot x^*$. By making $\ceil{x_e}$ parallel
copies of each edge $e$, one can view $x$ as yielding a feasible solution to 
(\kecsslp[{(k+2)}]),
and then apply Theorem~\ref{kecss-thm}. As described, this yields a pseudopolynomial-time
algorithm, but one can refine this to obtain a polynomial time algorithm, as described in
Section~\ref{kecsm}. 

\vspace*{-1ex}
\paragraph{Degree-bounded \kecss and \kecsm.}
We showcase the versatility of our techniques by showing in Section~\ref{degbnd-kecss},
that they readily extend to yield strong guarantees for the {\em degree-bounded
extensions of \kecss and \kecsm}, wherein we are also given node degree lower- and upper-
bounds $\{\ldeg_v,\udeg_v\}_{v\in V}$, and we seek a minimum-cost (\kecss or \kecsm)
solution satisfying these degree bounds.
Our results for \kecss and \kecsm extend to yield the same cost and connectivity
guarantees with (roughly) an additive violation of $2$ for the degree bounds. For \kecss,
we obtain in polytime a subgraph $H$, where $c(H)$ is at most the 
the optimal value $\lpopt[{\mdkecsslp}]$ of the LP-relaxation of the problem, 
we have $\ldeg_v-2\leq|\dt_H(v)|\leq\udeg_v+2$ for all $v\in V$,
and $H$ is $(k-2)$-edge connected for even $k$, and $(k-3)$-edge connected for odd $k$
(Theorem~\ref{mdkecss-thm}).
For \kecsm, we obtain a $k$-edge connected multigraph $H$ with 
$c(H)\leq\rho_k\text{(optimum)}$ and $\ldeg_v-2\leq|\dt_H(v)|\leq\rho_k\cdot\udeg_v+2$ for
all $v\in V$ (Theorem~\ref{mdkecsm-thm}). (Recall that $\rho_k=\bigl(1+\frac{2}{k}\bigr)$
if $k$ is even, and $\bigl(1+\frac{3}{k}\bigr)$ if $k$ is odd.) 

These are the first results for degree-bounded $\{\kecss,\kecsm\}$ where the cost of the
solution returned is at most the optimum, but the connectivity constraints are violated by 
an additive constant. Prior work on degree-bounded network design, when specialized to the
setting of degree-bounded \kecss, yields results where the solution $H$ obtained is
$k$-edge connected, but $c(H)$ is roughly $2\cdot\lpopt[{\mdkecsslp}]$, and there is some
violation of the degree constraints. The state-of-the-art for the degree violation in
these works is $\ldeg_v\leq|\dt_H(v)|\leq 2\udeg_v+3$ for all $v\in V$~\cite{LauNSS09},
and $|\dt_H(v)|\leq\min\{\udeg_v+3k,2\udeg_v+2\}$ for all $v\in V$, when only degree upper
bounds are present~\cite{LauZ15}; Section~\ref{relwork} discusses prior work
on degree-bounded network design in more detail.

\subsection{Technical contributions and overview}
We now provide a high-level overview of our algorithms and analyses. 
For any $S\sse V$ and $Z\sse E$, let $\dt_Z(S)$ denote $\dt(S)\cap Z$.
Our algorithms build on the
iterative-rounding technique introduced by Jain~\cite{Jai01} and further refined by Singh
and Lau~\cite{SL15}. 
We discuss first the algorithm leading to Theorem~\ref{kecss-thm}.
We begin by computing an extreme-point solution $\hx$ to \eqref{kecss-lp}.
Say that a set $S\sse V$, or its cut-constraint \eqref{con}, is {\em tight} (under
$\hx$), if $\hx\bigl(\dt(S)\bigr)=k$. 
If any variables are set to $1$ in $\hx$, we fix them to $1$ permanently, that is, we pick
the corresponding edges in our solution. 
Otherwise, we
look for a cut-constraint \eqref{con} that is nearly satisfied by the edges already
picked. For \kecss, this corresponds to identifying a set \( S \) with close to \( k \)
edges crossing it that have been fixed to $1$. We would like to then drop the cut-constraint
\eqref{con} for set \( S \), compute an extreme-point solution to the resulting LP, and
iterate.    

In the first iteration, before any constraints are dropped, it is straightforward to find
a set with at least \((k-2)\) edges of value $1$ crossing it. This
follows from the token-counting argument of Jain~\cite{Jai01} and crucially relies on the
existence of a laminar family of tight cut-constraints of full rank: by this,
we mean a collection $\Lc$ of tight node-sets 
such that 
any two $A,B\in\Lc$ satisfy $A\sse B, B\sse A$, or
$A\cap B=\es$ (laminar family), and, letting $F=\{e: 0<\hx_e<1\}$, the vectors
$\bigl\{\chi^{\dt_F(A)}\bigr\}_{A\in\Lc}$ form a (linearly-independent) basis of 
$\spn\bigl(\bigl\{\chi^{\dt_F(S)}\bigr\}_{S\text{ is tight}}\bigr)$ (full rank).
However, as the cut-constraints for certain sets are dropped in subsequent iterations, it
becomes unclear whether such a 
laminar family of tight constraints even exists. 
Hershkowitz, Klein, and Zenklusen~\cite{hershkowitz2024ghost} also notice this difficulty
and come up with the construct of 
adding fake ``ghost'' edges to overcome it. In particular, they show that if they cannot
construct a laminar family of tight constraints with full rank, they can make progress by
contracting certain sets and adding ghost edges, a technique they term 
\emph{ghost value augmentation}.  

{\em Our chief technical insight is that even after dropping cut-constraints that are
nearly satisfied, we can still find a full-rank, laminar family of tight cut-constraints}
(Lemma~\ref{laminar}). The technical engine that leads to this laminar structure is 
a key {\em ``uncrossing'' result} that we establish, showing that even after dropping
cut-constraints that are nearly satisfied, tight sets can be uncrossed
(Lemma~\ref{lcspan}). Uncrossing tight sets means that for any two not-dropped tight
sets $A,B$, 
we can come up with a laminar collection of not-dropped tight sets $\Lc$ such that
$\bigl\{\chi^{\dt_F(A)},\chi^{\dt_F(B)}\bigr\}\sse\spn\bigl(\bigl\{\chi^{\dt_F(S)}\bigr\}_{S\in\Lc}\bigr)$. 
With this key uncrossing property in hand, we can utilize standard arguments to find the
desired laminar family, and use the same token-counting argument of~\cite{Jai01} to find a
set whose cut-constraint is nearly satisfied, drop this set, and iterate.

We point out that establishing the uncrossing property is the place where we completely
diverge from the approach in~\cite{hershkowitz2024ghost}. 
Hershkowitz et al.~\cite{hershkowitz2024ghost} identify the potential lack of
uncrossability after dropping sets, as the barrier toward making standard approaches
work,%
\footnote{Indeed, they state that ``The major barrier to
the standard iterative relaxation approach is that it does not appear that uncrossing is
possible.''} 
and (as noted above) come up with the 
novel workaround of ghost edges and ghost-value augmentation as a means of restoring
uncrossability and thereby recovering the laminar structure of tight cut-constraints.
In contrast, and contrary to the sentiment expressed in~\cite{hershkowitz2024ghost},
our key insight is that {\em tight cut-constraints can still be uncrossed
(as is)},
without 
the need for any workaround. 
This is the source of our savings, compared to~\cite{hershkowitz2024ghost}, both in
the additive approximation factor, where our bounds are nearly optimal, and 
the simplicity of the underlying algorithm and analysis. 
For even $k$, our algorithm (see Algorithm~\ref{kecss-alg}) is simply the
following: 
(1) compute an extreme-point solution to the current LP; 
(2) pick edges whose variables are set to $1$, 
(3) drop sets crossed by at least $(k-2)$ picked edges, 
(4) update the requirements, i.e., the right-hand-sides (RHSs) of constraints \eqref{con}
for the remaining sets to take into account the edges already picked; 
{\sc Iterate.}
The laminar structure of tight cut constraints (which follows due to uncrossing) enables
one to easily argue that there is always some new edge that is picked, and some set that
is dropped, in every iteration. 
(As a byproduct, we point out that this proves the so-called {\em light cut property}
in~\cite{hershkowitz2024ghost}, which is stated as a ``very interesting open
problem.'') 

We prove that the uncrossing property holds for any cut-requirement function $f:2^V\mapsto\Z$
(specifying the RHSs of constraints \eqref{con}) that is {\em two-way uncrossable} and has
{\em even parity}%
\footnote{ 
Two-way uncrossable means that we have
$f(A)+f(B)\leq\min\bigl\{f(A\cap B)+f(A\cup B),\,f(A-B)+f(B-A)\bigr\}$ 
for any $A,B\sse V$ such that $A\cap B, A-B, B-A, V-(A\cup B)$ are all non-empty.  
Even parity means that $f(A)+f(B)+f(A\cup B)$ is even for all $A,B\neq\es$ and disjoint.} 
(see Definition~\ref{crprops}), properties that we crucially exploit in our proof. 
We believe that this is a tool of independent interest. 
Observe that, for even $k$, these properties hold at the beginning, when
$f=f^{\kecss}$ (where $f^{\kecss}(S)=k$ if $\es\neq S\subsetneq V$, and is $0$ otherwise),
and they continue to hold when $f$ is updated after picking edges (see Lemma~\ref{twoway}
(b) and Claim~\ref{even}). 

To obtain $(k-1)$-edge connectivity incurring cost at most $1.5\cdot\lpopt$
(Theorem~\ref{bicriteria-kecss} and Section~\ref{bicriteria}), we tweak the
above approach. We now drop a set only when it is crossed by at least $(k-1)$ picked
edges. We argue that the uncrossing property still holds with this more restrictive dropping
rule; see Lemma~\ref{new-lcspan}. Interestingly, to establish this, we only need that the
cut-requirement function $f$ is two-way uncrossable (but it need not have even parity). So
we can again obtain a laminar family of tight cut-constraints. It need not now be the case
that there is always an integral edge on the boundary of some set in the laminar family;
however, we can show that there is always some edge with $\hx_e\geq\frac{2}{3}$. So we
pick such edges, drop sets as described above, and iterate, to obtain the above
guarantee. With unit costs, one can perform a tighter cost analysis using the fact that
$\lpopt\geq\frac{kn}{2}$ in this case.

Finally, to handle the degree-bounded versions of \kecss and \kecsm
(Section~\ref{degbnd-kecss}), we proceed quite similarly to the non-degree-bounded
versions. Essentially, everything works out as before, since singleton sets corresponding
to tight node-degree constraints can be added to any laminar family while preserving
laminarity. So we can again pick integral edges, and when we pick such edges we also
modify the node degree bounds and drop the degree constraints for a node if they are close
enough to being satisfied.  

\medskip
We believe that the insights developed in this work will be valuable in obtaining similar
results for more-general network design problems. In particular, one enticing question
that arises is whether similar additive guarantees are possible for more-general survivable
network design problems.

\subsection{Other related work} \label{relwork}
\kecss and its special cases have been studied extensively in the literature. 
One special case that has received much attention is unweighted \kecss (and unweighted
\kecsm), i.e., where all edges have cost $1$.
When \( G \) is simple, 
Cheriyan and Thurimella~\cite{CT00} presented a \((1 + 2/k)\)-approximation algorithm
using a purely combinatorial approach.  
Gabow, Goemans, Tardos, and Williamson~\cite{GGTW09}, used the iterative-rounding framework of
Jain~\cite{Jai01} to obtain a $(1+2/k)$-approximation algorithm for unweighted \kecss even
on multigraphs, and Gabow and Gallagher~\cite{GabowG12} improved the approximation factor
to $\bigl(1+\frac{1.91}{k}\bigr)$ (and $\bigl(1+\frac{1.89}{k}\bigr)$ for even $k$), which is
the current state-of-the art for this problem.
The results of~\cite{GabowGTW09,GabowG12}, which compare the cost against the LP-optimum,
can also be used to yield guarantees of the form stated in Theorem~\ref{kecss-thm}, 
where one obtains cost at most the optimum by
sacrificing some connectivity. This is because one can easily see that
$\lpopt[{\kecsslp[{(k-2)}]}]\leq\frac{k-2}{k}\cdot\lpopt$, and so (for instance) the result
of Gabow et al.~\cite{GabowGTW09} for \kecss[{(k-2)}] yields a $(k-2)$-edge connected
subgraph of cost at most
$\bigl(1+\frac{2}{k-2}\bigr)\cdot\lpopt[{\kecsslp[{(k-2)}]}]\leq\lpopt$.
On the hardness side, Gabow et al.~\cite{GabowGTW09} showed that, for all $k\geq 2$, there
is some absolute constant $c>0$ such that unweighted \kecss does not
admit a $\bigl(1+\frac{c}{k}\bigr)$-approximation, unless \p = \np. 
After our work was posted on the arXiv~\cite{KumarS25v1}, Nutov posted a manuscript on the
arXiv~\cite{Nutov25} obtaining a subset of the results obtained here for \kecss and
\kecsm, acknowledging that these were obtained subsequent to our work.

Unlike \kecss (with general edge costs), for which improving upon the long-standing
2-approximation algorithm remains an open problem, \kecsm admits better approximation
algorithms. For \kecsm, 
Frederickson and J\'aJ\'a~\cite{FJ81,FJ82} gave a \(3/2\)-approximation algorithm for even
\( k \) and a \((3/2 + O(1/k))\)-approximation for odd \( k \). Karlin, Klein, Oveis
Gharan, and Zhang~\cite{KKOZ22} made a significant improvement, and gave a \((1 +
O(1/\sqrt{k}))\)-approximation algorithm for \kecsm. 
Most recently, Hershkowitz et al.~\cite{hershkowitz2024ghost} gave an asymptotically tight
\((1 + O(1/k))\)-approximation algorithm for \kecsm.   

There has been significant recent progress on \kecsm for $k = 2$, largely due to its
connection with the metric traveling salesman problem (metric-TSP). In a recent
breakthrough, Karlin, Klein, and Oveis Gharan~\cite{KKO21a, KKO22} presented a better than
$1.5$ approximation for metric-TSP, which also implies a better-than-$1.5$ approximation
algorithm for $\kecsm[2]$. 

Iterative rounding is a powerful technique introduced by Jain~\cite{Jai01}, that involves
iteratively fixing variables by exploiting structural properties of extreme-point optimal
solutions. Later work of~\cite{SL15,LauNSS09} on degree-bounded network design added
another ingredient to this technique, namely iteratively dropping some constraints. 
This paradigm of 
{\em iterative rounding and relaxation} has proved to quite successful in developing
algorithms for network-design problems, and more generally in combinatorial optimization;
see~\cite{LRS11} for an extensive study. 
In particular, the state-of-the-art for the survivable network-design problem (\sndp) with
node degree bounds is achieved via iterative rounding and relaxation. 
For general \sndp, or even the special case of \kecss, the only work 
that we are aware of that considers both lower- and upper-bounds on node degrees is due to  
Lau et al.~\cite{LauNSS09}.
They 
present a polytime algorithm that outputs a subgraph $H$ satisfying all the connectivity
requirements such that $c(H)$ is at most twice the LP optimum, 
and $\ldeg_v\leq|\dt_H(v)|\leq 2\udeg_v+3$ for all $v\in V$.
When only degree upper bounds are present, the degree-violation above was
improved by~\cite{LouisV10} and then~\cite{LauZ15} to 
$|\dt_H(v)|\leq\min\{\udeg_v+3r_{\max},2\udeg_v+2\}$ for all $v\in V$ (with the same bound on
$c(H)$), where $r_{\max}$ is the maximum edge-connectivity requirement between any pair of
nodes.  
When $G$ is the complete graph, the edge costs from a metric, and with only
degree upper bounds, Chan et al.~\cite{ChanFLY11}, devised a polytime algorithm that
returns a feasible solution (i.e., where all connectivity {\em and} degree constraints are
satisfied) of cost at most $O(1)$ times the LP-optimum. In particular, for degree-bounded
\kecss, they obtain a $k$-edge connected subgraph satisfying the degree bounds of cost
at most $\bigl(1+\frac{1}{k}\bigr)$ times the LP-optimum. 
Finally, we note that for the degree-bounded MST problem, Singh and Lau~\cite{SL15}
obtained a tight result, producing a spanning tree of cost at most the optimum with an
additive violation of $1$ for the degree (lower and upper) bounds.

We note that, with the notable exception of~\cite{hershkowitz2024ghost}, almost all
applications of iterative rounding and relaxation for
network-design problems proceed by not touching the network-connectivity constraints,
in order to retain the uncrossing properties.

\section{Preliminaries and notation} \label{prelim}
Recall that we are given an undirected graph $G=(V,E)$ with nonnegative edge costs
$\{c_e\}_{e\in E}$. For a subset $S\sse V$ and subset $Z\sse E$, which we will
interchangeably view as the subgraph $(V,Z)$ of $G$, we use $\dt_Z(S)$ to denote
$\dt_G(S)\cap Z$, i.e., the edges of $Z$ on the boundary of $S$.
For $x\in\R^E$, we use $\supp(x):=\{e:E: x_e>0\}$ to denote the support of $x$.

We say that two subsets $A,B\sse V$ are {\em crossing}, if all four sets 
$A\cap B, A-B, B-A, V-(A\cup B)$ are non-empty; we sometimes say that $A, B$ are 
{\em weakly-crossing} if $A\cap B, A-B, B-A$ are all non-empty. A {\em laminar family}
$\Lc$ of sets denotes a collection of sets where no two sets in the collection are
weakly-crossing.

\vspace*{-1ex}
\paragraph{Cut-requirement functions.}
It will be useful to specify \kecss (and \kecsm) using the framework of 
{\em cut-requirement functions}. 
A cut-requirement function is a function $f:2^V\mapsto\Z$. This can be used to specify the
network-design problem where we seek a minimum-cost solution $H$ such that 
$|\dt_H(S)|\geq f(S)$ for all $S\sse V$.
Cut-requirement functions constitute a very versatile framework for specifying
network-design problems: for instance, \kecss corresponds to the cut-requirement 
function $f^{\kecss}$ where $f^{\kecss}(S):=k$ for all $\es\neq S\subsetneq V$, and
$f^{\kecss}(\es)=f^{\kecss}(V)=0$; 
various other network-design problems, such as the $T$-join problem, point-to-point
connection problem etc., can be modeled using suitable cut-requirement functions 
(see, e.g.,~\cite{goemans1995general, goemans1997primal}). We define below various properties of
cut-requirement functions. 

\begin{definition} \label{crprops}
Let $f:2^V\mapsto\R$. We say that:
\begin{enumerate}[label=$\bullet$, topsep=0.2ex, noitemsep, leftmargin=*]
\item $f$ is {\em symmetric} if $f(S)=f(V-S)$ for all $S\sse V$; 
$f$ is {\em normalized} if $f(\es)=f(V)=0$
\item $f$ is {\em crossing supermodular} if for any crossing sets $A,B\sse V$, we have  
$f(A)+f(B)\leq f(A\cap B)+f(A\cup B)$
\item $f$ is {\em two-way uncrossable} if for any crossing sets $A,B\sse V$, we have 
$f(A)+f(B)\leq\min\bigl\{f(A\cap B)+f(A\cup B),\,f(A-B)+f(B-A)\bigr\}$. 
\item $f$ is {\em weakly supermodular} if for any sets $A,B\sse V$, we have 
$f(A)+f(B)\leq\max\bigl\{f(A\cap B)+f(A\cup B),\,f(A-B)+f(B-A)\bigr\}$. 
\end{enumerate}

\noindent
For $f:2^V\mapsto\Z$, we say that $f$ has {\em even parity} if for any non-empty disjoint 
sets $A,B\sse V$, we have that $f(A)+f(B)+f(A\cup B)$ is even.
\end{definition}

Some relationships between these properties are easy to infer. Clearly, if $f$ is
two-way uncrossable, then it is also crossing supermodular.
Lemma~\ref{twoway} shows that the converse is also true, if $f$ is symmetric.
Also, if $f$ is symmetric and two-way uncrossable, then it is weakly supermodular: if
sets $A,B\sse V$ do not cross, then we have $\{A-B,B-A\}=\{A,B\}$, or 
$\{A\cap B,A\cup B\}=\{A,B\}$, or $\{V-A,V-B\}=\{A-B,B-A\}$.

In our algorithm, we iteratively pick edges, and modify the cut-requirement function
correspondingly. 
Lemma~\ref{twoway} (b) and Claim~\ref{even} show respectively that this preserves two-way
uncrossability and even parity.  

\begin{lemma} \label{twoway}
Let $f:2^V\mapsto\R$. 
\begin{enumerate}[label=(\alph*), topsep=0.1ex, noitemsep, leftmargin=*]
\item If $f$ is symmetric and crossing supermodular, then it is two-way uncrossable.
\item Let $w:E\mapsto\R_+$. If $f$ is two-way uncrossable, then the function
$g:2^V\mapsto\R$ defined by $g(S):=f(S)-w\bigl(\dt(S)\bigr)$ for all $S\sse V$, is also
two-way uncrossable.
\end{enumerate}
\end{lemma}

\begin{proof}
For part (a), we only need to argue that $f(A)+f(B)\leq f(A-B)+f(B-A)$. This follows by
noting that $f(B)=f(V-B)$, applying the crossing-supermodularity property to sets 
$A, V-B$, which are also crossing, and noting that $f\bigl(A\cap (V-B)\bigr)=f(A-B)$ and  
$f\bigl(A\cup (V-B)\bigr)=f\bigl(V-(B-A)\bigr)=f(B-A)$.  

Part (b) follows because the cut function $w\bigl(\dt(S)\bigr)$ is symmetric and
submodular, and so for any two sets $A,B\sse V$, we have
\[
w\bigl(\dt(A)\bigr)+w\bigl(\dt(B)\bigr)\geq\max\Bigl\{
w\bigl(\dt(A\cap B)\bigr)+w\bigl(\dt(A\cup B)\bigr),\,
w\bigl(\dt(A-B)\bigr)+w\bigl(\dt(B-A)\bigr)\Bigr\}. \qedhere
\]
\end{proof}

\begin{claim} \label{even}
Let $f:2^V\mapsto\Z$ have even parity, and $H\sse E$. Define $g:2^V\mapsto\Z$ by
$g(S):=f(S)-|\dt_H(S)|$ for all $S\sse V$. Then, $g$ has even parity.
\end{claim}

\begin{proof}
Let $A,B\sse V$ be non-empty disjoint sets. Then,
\[
f(A)+f(B)+f(A\cup B)-g(A)-g(B)-g(A\cup B)=|\dt_H(A)|+|\dt_H(B)|+|\dt_H(A\cup B)|
=2|\dt_H(A\cup B)|+2\tht
\]
where $\tht$ is the number of edges with one endpoint in $A$ and other endpoint in $B$. 
So $f(A)+f(B)+f(A\cup B)$ and $g(A)+g(B)+g(A\cup B)$ have the same parity; so since $f$
has even parity, $g$ has even parity.
\end{proof}

Finally, we note that since we are working on undirected graphs, we may assume that the
cut-requirement function $f$ is symmetric and normalized. 
In particular, if $f$ is 
two-way uncrossable, then we can always
``inject'' symmetry and normalization while maintaining this property,
and without changing the underlying network-design problem. 
The following result, proved in Appendix~\ref{append-prelim}, makes this precise.

\begin{lemma} \label{symnormal}
Let $f:2^V\mapsto\R$, and $x\in\R_+^E$. 
Define the symmetrization and normalization of $f$ to be the function $g:2^V\mapsto\R$
defined by $g(S):=\max\{f(S),f(V-S)\}$ for all $\es\neq S\subsetneq V$, and $g(\es)=g(V)=0$.
\begin{enumerate}[label=(\alph*), topsep=0.1ex, noitemsep, leftmargin=*]
\item If $f$ is two-way uncrossable, then so is $g$. \label{symtwoway}
\item We have $x\bigl(\dt(S)\bigr)\geq f(S)$ for all $\es\neq S\subsetneq V$ iff
$x\bigl(\dt(S)\bigr)\geq g(S)$ for all $S\sse V$. \label{symndp}
\item Suppose that $x\bigl(\dt(S)\bigr)\geq f(S)$ for all $\es\neq S\subsetneq V$.
Let $T\sse V$. If $x\bigl(\dt(T)\bigr)=f(T)$ then $x\bigl(\dt(T)\bigr)=g(T)$.
\label{symtight}
\end{enumerate}
\end{lemma}

\section{Algorithm for \boldmath \kecss} \label{kecss}
We now describe and analyze the algorithm 
for \kecss that yields a $(k-2)$-edge connected subgraph of cost at most $\lpopt$, for
even $k$, thereby proving Theorem~\ref{kecss-thm}. In Section~\ref{bicriteria}, we show
how the the algorithm and analysis can be modified to obtain a $(k-1)$-edge connected
subgraph of cost at most $1.5\cdot\lpopt$ (for all $k$).
Recall the following LP-relaxation for \kecss.
\begin{alignat}{3}
\min & \quad & \sum_{e\in E}c_ex_e \tag{\kecsslp} \\
\text{s.t.} \quad && x\bigl(\dt(S)\bigr) & \geq k
\qquad && \forall \es\neq S\subsetneq V \tag{\ref{con}} \\
&& 0 \leq x_e & \leq 1 && \forall e\in E. \tag{\ref{bnd}}
\end{alignat}
Throughout this section, we assume that $k$ is even, since the case of odd $k$ can be
handled by considering $k-1$, which is even. 
Note that since $k$ is even, $f^{\kecss}$ has even parity. 
In terms of cuts, our goal is to produce a subgraph $H$ of cost at most $\lpopt$
such that $|\dt_H(S)|\geq k-2$ for all $\es\neq S\subsetneq V$.

\paragraph{An overview.}
We first give an overview of our algorithm and analysis. Our algorithm is actually quite
simple. As noted earlier, we iteratively pick edges. 
Letting $H$ denote the currently picked edges, we consider the {\em residual LP}
\eqref{reslp}, which is of the same form as \eqref{kecss-lp}, 
but with the RHS of constraints \eqref{con} replaced by the residual cut-requirement
function given by $f^\res(S):=f^{\kecss}(S)-|\dt_H(S)|$ for all $S\sse V$. 
Note that by Lemma~\ref{twoway} (b), $f^\res$ is two-way uncrossable, and by
Claim~\ref{even}, $f^\res$ has even parity.

We compute an extreme-point optimal solution $\hx$ to this LP.
It is well known that when the RHS of constraints \eqref{con} is given by a
weakly-supermodular function, for any extreme-point solution to \eqref{kecss-lp}, one can  
associate a laminar family of tight constraints \eqref{con} that uniquely defines the
extreme point (see, e.g.,~\cite{Jai01}).
We exploit this laminar structure to argue that if all sets of this laminar family have
residual requirement 
at least $3$, then $\hx$ must contain some integral edge (i.e., an edge $e$ with
$\hx_e=1$) and there is some tight set $T$ from the laminar family with at most $3$
fractional edges crossing it; see Lemma~\ref{largereq}. We pick the integral edges, and
note that after doing so the residual requirement of $T$ becomes at most $2$.

We have $|\dt_H(T)|\geq k-2$, and so we have achieved our goal for set $T$. The
natural thing to do now is to drop set $T$ from the collection of constraints
\eqref{con}. This is indeed what we do: we drop {\em all} sets whose residual
requirement is at most $2$ (see step~\ref{kecss-update}). This however means that the RHS 
of constraints \eqref{con} is no longer given by a weakly-supermodular function. (That is,
if we treat the requirement of dropped sets as $0$, then the resulting cut-requirement
function is not weakly supermodular.)
Nevertheless, 
{\em we prove that any extreme-point solution to the resulting LP can still be defined via
a laminar family of tight constraints} (Lemma~\ref{laminar}). This is our 
{\em key technical insight}, and the result that drives our algorithm and analysis.
We show that this property holds for any cut-requirement function that is two-way
uncrossable and has even parity, and we believe that this result is of independent
interest and may have applications beyond the current context.

\begin{procedure}[t!] 
\caption{$k$-ECSS-Alg() \textnormal{\qquad\qquad // LP-rounding algorithm for \kecss}
\label{kecss-alg}} 
\KwIn{\kecss instance $\bigl(G=(V,E),\{c_e\geq 0\}_{e\in E},\,\text{\em even}\ k\geq 4\bigr)$}
\KwOut{subset $H\sse E$}

Initialize $E'\assign E$, $H\assign\es$, $f^\res=f^{\kecss}$ (i.e., $f^\res(S)=k$ for all
$\es\neq S\subsetneq V$, $f^\res(\es)=f^\res(V)=0$). 

Let $\Sc=\{S\sse V: f^\res(S)\geq 3\}$.

\While{$H$ is not $(k-2)$-edge connected (equivalently $\Sc\neq\es$)}{
Compute an extreme-point optimal solution $\hx\in\R^{E'}$ to the following residual LP:
\begin{alignat}{3}
\min & \quad & \sum_{e\in E'}c_ex_e \tag{ResLP} \label{reslp} \\
\text{s.t.} \quad 
&& x\bigl(\dt_{E'}(S)\bigr) & \geq f^\res(S) \qquad && \forall S\in\Sc \label{rescon} \\
&& 0 \leq x_e & \leq 1 && \forall e\in E'. \notag
\end{alignat}

Update $H\assign H\cup\{e\in E':\hx_e=1\}$, $E'\assign E'-H$,
$f^\res(S)=f^{\kecss}(S)-|\dt_H(S)|$ for all
$S\sse V$, and $\Sc=\{S\sse V: f^\res(S)\geq 3\}$. \label{kecss-update}
}

\Return $H$.
\end{procedure}

Given this, we can simply iterate the above process. We have a laminar family where all
residual requirements are at least $3$, by {\em construction}, since we have dropped the
constraints of all other sets. So, as stated earlier, by Lemma~\ref{largereq}, we are
guaranteed to find an integral edge, and make progress. 

Given Lemmas~\ref{laminar} and~\ref{largereq}, which show that the algorithm is
well-defined, the proof of the performance
guarantee of the algorithm is fairly immediate. 
We obtain $c(H)\leq\lpopt$, since we only pick edges that are set to $1$ by an
optimal solution, and we obtain that $H$ is at least $(k-2)$-edge connected, since we only
drop sets with residual requirement at most $2$. 
The only thing left to argue is that one can solve the
residual LP \eqref{reslp} in polynomial time. We show that one can devise a polytime
separation oracle for the residual LP (Lemma~\ref{lpsolve}). For \eqref{kecss-lp} (i.e.,
the starting LP), 
such an oracle amounts to a min-cut computation. For the residual LP, we argue that this
amounts to an approximate min-cut computation, where we only need to consider
$2$-approximate min-cuts. 
This is because 
if $x$ is feasible for the residual LP, then the min-cut value with edge weights $x_e$ or
$1$ respectively depending on whether $e$ belongs to the residual graph or $H$, is at least
$k-2$. Since $2(k-2)\geq k$ for $k\geq 4$, this means that if $x$ is infeasible, then
there is a $2$-approximate min-cut $S$ of value less than $k$, where $S$ is not a dropped 
set, i.e., $|\dt_H(S)|<k-2$.
We can check if the min-cut value is at least $k-2$, and if so, enumerate the
$2$-approximate min-cuts and check that the LP-constraint is met for these cuts. 
There are at most $O(n^4)$ $2$-approximate min-cuts~\cite{Kar99}, and we can enumerate
all these cuts in polytime~\cite{NNI94}.

\paragraph{Analysis.}
We begin by arguing that \eqref{reslp} can be solved efficiently.

\begin{lemma} \label{lpsolve}
There is a polytime separation oracle for \eqref{reslp}.
\end{lemma}

\begin{proof}
Let $x\in\R^{E'}$. We can easily check if $0\leq x_e\leq 1$ for all $e\in E'$ in
polytime. So assume this holds. Define $w_e=x_e$ if $e\in E'$, and $w_e=1$ if 
$e\in H$. So $x$ is feasible to \eqref{reslp} iff $w\bigl(\dt_G(S)\bigr)\geq k$ for all
$S\in\Sc$. Note that, by construction, we have $w\bigl(\dt(S)\bigr)\geq k-2$ for all 
$S\notin\Sc$, $\es\neq S\subsetneq V$.

So we check that the min-cut with $\{w_e\}_{e\in E}$ edge capacities is at least $k-2$;
if not, constraint \eqref{rescon} for the min-cut $S$ is a violated inequality. Since
$k\geq 4$, we have $2(k-2)\geq k$, and so any $S$ with $w\bigl(\dt(S)\bigr)<k$ is a
$2$-approximate min-cut. So we next enumerate all $2$-approximate min-cuts:
there are at most $O(n^4)$ such cuts~\cite{Kar99}, and they can enumerated in
polynomial time~\cite{NNI94,chekuri2020lp,beideman2023approximate}. If $w\bigl(\dt(S)\bigr)<k$
and $S\in\Sc$, for any of these cuts, then we can return \eqref{rescon} for set $S$ as a
violated constraint; otherwise $x$ is feasible to \eqref{reslp}. 
\end{proof}

Next, we prove some results that will show that Algorithm~\ref{kecss-alg} is well-defined
and makes progress in each iteration. We prove that there always exists an edge $e$ in
our current edge-set $E'$ with $\hx_e=1$, and that the collection $\Sc$ strictly
reduces in each iteration. 
To this end, we prove that any extreme point of \eqref{reslp}
is defined by a laminar family of tight sets (Lemma~\ref{laminar}). 
As mentioned earlier, this is our key technical insight, showing that laminarity holds
despite the fact that we have dropped some cut-constraints (corresponding 
to $S\sse V$ with $f^\res(S)\leq 2$). We prove this more generally, whenever the RHS of
\eqref{rescon} is given by a two-way uncrossable function of even parity, and we believe
that this statement is of independent interest.
The chief ingredient of this proof is encapsulated by
Lemma~\ref{lcspan}, which shows that tight sets can be {\em uncrossed}; here, we crucially
exploit that two-way uncrossability and even parity.
Capitalizing on this laminar structure, Lemma~\ref{largereq} establishes the required
integrality property. Lemma~\ref{largereq} argues the contrapositive: if all edges are
fractional, then the RHS of some constraint must be at most $2$.

\begin{lemma}[{\bf Uncrossing tight sets}] \label{lspan} \label{lcspan}
Let $f:2^V\mapsto\Z$ be a symmetric, normalized, two-way uncrossable, even-parity function.
Let $\Sc:=\{S\sse V: f(S)\geq 3\}$.
Let $x\in\R_+^E$ satisfy $x\bigl(\dt(S)\bigr)\geq f(S)$ for all $S\in\Sc$.
Let $Z=supp(x)$. 
Say that a set $S\sse V$ is tight if $x\bigl(\dt(S)\bigr)=f(S)$.
Let $A,B\in\Sc$ be two weakly-crossing tight sets. 
Then, there exists a laminar family $\Lc\sse\Sc\cap\{A,A\cap B,A\cup B,A-B,B-A\}$ of tight
sets with $A\in\Lc$ such that 
$\chi^{\dt_Z(B)}\in\spn\bigl(\bigl\{\chi^{\dt_Z(S)}\bigr\}_{S\in\Lc}\bigr)$.
\end{lemma}

\begin{proof}
Since $f$ is symmetric, if $S\in\Sc$, we also have that $V-S\in\Sc$. 
Define $g:2^V\mapsto\R$ by
$g(S)=f(S)-x\bigl(\dt(S)\bigr)$ for all $S\sse V$.
We may assume that $A$ and $B$ are crossing, since otherwise we have $A\cup B=V$, and so
$A-B=V-B\in\Sc$, and we can take $\Lc=\{A-B,A\}$.
Note that at least one of $A\cap B$ and $A\cup B$ is in $\Sc$, since 
$f(A\cap B)+f(A\cup B)\geq f(A)+f(B)\geq 3+3=6$. Similarly, since $f$ is two-way
uncrossable, at least one of $A-B$ and $B-A$ is in $\Sc$. 

Let $Z'=\{uv\in Z: u\in A-B, v\in B-A\}$ and $\tht=x(Z')$.
Let $Z''=\{uv\in Z: u\in A\cap B, v\notin A\cup B\}$ and $\gm=x(Z'')$
(see Fig.~\ref{tightuncross}). 
If $x\bigl(\dt(S)\bigr)\geq f(S)$ for both $S=A\cap B$ and $S=A\cup B$,  
then from standard arguments, we can show that $A\cap B$ and $A\cup B$ are tight, and
$\tht=0$. We have  
\begin{equation*}
0=f(A)-x\bigl(\dt(A)\bigr)+f(B)-x\bigl(\dt(B)\bigr)\leq 
f(A\cap B)-x\bigl(\dt(A\cap B)\bigr)+f(A\cup B)-x\bigl(\dt(A\cup B)\bigr)-2\tht\leq-2\tht\leq 0
\end{equation*}
where the first inequality is due to the two-way uncrossability of $g$, and the second is 
because 
$x\bigl(\dt(S)\bigr)\geq f(S)$ for $S\in\{A\cap B,A\cup B\}$.
It follows that we must have equality throughout, so $x\bigl(\dt(S)\bigr)=f(S)$ for
$S\in\{A\cap B,A\cup B\}$ and $\tht=0$. 
This implies that 
$\chi^{\dt_Z(A)}+\chi^{\dt_Z(B)}=\chi^{\dt_Z(A\cap B)}+\chi^{\dt_Z(A\cup B)}$.
So if both $A\cap B, A\cup B$ are in $\Sc$, then we can take $\Lc=\{A,A\cap B,A\cup B\}$, 
proving the lemma. 
 
Similarly, if $x\bigl(\dt(S)\bigr)\geq f(S)$ for both $S=A-B$ and $S=B-A$, then again, 
using the two-way uncrossability of $g$, we obtain that $A-B$ and $B-A$ are
tight, and $\gm=0$.
So $\chi^{\dt_Z(A)}+\chi^{\dt_Z(B)}=\chi^{\dt_Z(A-B)}+\chi^{\dt_Z(B-A)}$, and 
therefore if both $A-B, B-A$ lie in $\Sc$, then we can take
$\Lc=\{A,A-B,B-A\}$. 

We are left with the case where exactly one of $A\cap B,A\cup B$ is in $\Sc$, and exactly
one of $A-B,B-A$ is in $\Sc$. We will argue that even then, {\em all four sets 
$A\cap B, A\cup B, A-B, B-A$ are tight, and $\tht=\gm=0$}. This coupled with one
additional insight will prove the lemma. 
There are four cases to consider, but once we prove the lemma in one case, the
other cases follow from symmetric arguments.
\begin{enumerate}[label=\arabic*), topsep=0.3ex, itemsep=0.1ex, leftmargin=*]
\item Suppose $A\cup B, A-B\in\Sc$. So $A\cap B\notin\Sc$, $B-A\notin\Sc$, that is,
$f(A\cap B), f(B-A)\leq 2$.  We first establish that
$x\bigl(\dt(A\cap B)\bigr)+x\bigl(\dt(B-A)\bigr)\geq f(A\cap B)+f(B-A)$. 
Suppose not. L
Let $Z_{A\cap B,B-A}=\{uv\in Z: u\in A\cap B, v\in B-A\}$ and $\al=x(Z_{A\cap B,B-A})$
(see Fig.~\ref{tightuncross}). 
Then, we have 
\begin{equation}
f(B)=x\bigl(\dt(B)\bigr)=x\bigl(\dt(A\cap B)\bigr)+x\bigl(\dt(B-A)\bigr)-2\al
<f(A\cap B)+f(B-A)-2\al. \label{lcineq2}
\end{equation}
Thus, $f(A\cap B)+f(B-A)>f(B)+2\al$. If $f(B)\geq 4$, then this yields a contradiction
since $f(A\cap B), f(B-A)\leq 2$. If $f(B)=3$, then we still obtain a contradiction,
because then \eqref{lcineq2} implies that $f(A\cap B)=2=f(B-A)$, and this contradicts that
$f$ has even parity. 

So we have $x\bigl(\dt(A\cap B)\bigr)+x\bigl(\dt(B-A)\bigr)\geq f(A\cap B)+f(B-A)$. 
So $x\bigl(\dt(S)\bigr)\geq f(S)$ for some $S\in\{A\cap B, B-A\}$.
But then we must have $x\bigl(\dt(S)\bigr)=f(S)$, by the discussion in the prior
paragraphs, since $A\cup B$ and $A-B$ are both in $\Sc$. This also means that if $T$ is
the set other than $S$ in $\{A\cap B,A-B\}$, then $x\bigl(\dt(T)\bigr)\geq f(T)$,
because otherwise inequality \eqref{lcineq2} would still hold, yielding a
contradiction. Again, $x\bigl(\dt(T)\bigr)\geq f(T)$ implies that we actually have
equality here. 

So the upshot is that all four sets $A\cap B, A\cup B, A-B, B-A$ are tight, which also
means that $\tht=\gm=0$. Now we claim that we also have $\al=0$. This is because, as in
\eqref{lcineq2}, we obtain that $f(B)=f(A\cap B)+f(B-A)-2\al$. If $\al>0$, then we must have 
$f(A\cap B)+f(B-A)\geq 4$ and $f(B)=3$. But then we have $f(B)=3$ and 
$f(A\cap B)=f(B-A)=2$, which contradicts that $f$ has even parity.

Finally, observe that given all this, we have
$\chi^{\dt_Z(B)}=\chi^{\dt_Z(A-B)}+\chi^{\dt_Z(A\cup B)}-2\chi^{\dt_Z(A)}$, so taking
$\Lc=\{A,A-B,A\cup B\}$ proves the lemma. See Fig.~\ref{tightuncross} for an illustration.

\begin{figure}[ht!] 
\centering
\includegraphics[width=3.5 in]{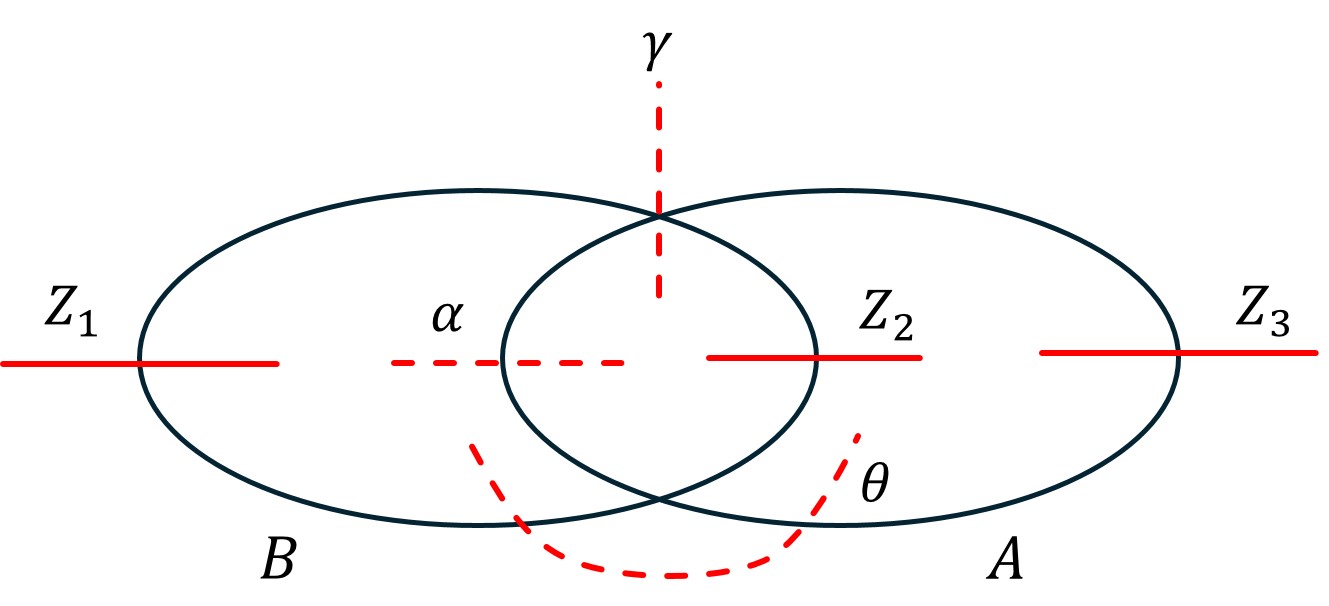}
\captionsetup{font=small, labelfont=small}
\caption{
  $Z_i$ denotes all edges in $Z$ with endpoints in the two sets containing the endpoints
  of the corresponding line. The quantities $\tht,\gm,\al$ denote the total $x$-weight of
  edges in $Z$ with endpoints in the two sets containing the endpoints of the
  corresponding arc.
  We prove that $\tht=\gm=0$, and in case 1) above, that
  $\al=0$.
  We can see then that $\chi^{\dt_Z(B)}=\chi^{Z_1}+\chi^{Z_2}$, 
  $\chi^{\dt_Z(A-B)}=\chi^{Z_2}+\chi^{Z_3}$, 
  $\chi^{\dt_Z(A\cup B)}=\chi^{Z_1}+\chi^{Z_3}$, and $\chi^{\dt_Z(A)}=\chi^{Z_3}$.
  \label{tightuncross}
}
\end{figure}

\item Suppose $A\cup B, B-A\in\Sc$. So $A\cap B\notin\Sc$, $A-B\notin\Sc$.
This case is completely symmetric to case 1), reversing the roles of $A$ and $B$. So we
obtain that $A\cap B, A\cup B, A-B, B-A$ are all tight sets, and 
$\chi^{\dt_Z(A)}=\chi^{\dt_Z(B-A)}+\chi^{\dt_Z(A\cup B)}-2\chi^{\dt_Z(B)}$. So we can take
$\Lc=\{A,B-A,A\cup B\}$.

\item Suppose $A\cap B, B-A\in\Sc$. So $A\cup B\notin\Sc$, $A-B\notin\Sc$. 
Let $\bB=V-B$. 
Note that the sets $A, \bB, A\cup\bB=V-(B-A), A-\bB=A\cap B$ are in
$\Sc$, and sets $A\cap\bB=A-B$ and $\bB-A=V-(A\cup B)$ are not in $\Sc$.
So considering the sets $A, \bB$, we are in case 1). Therefore, the sets 
$A\cap\bB, A\cup\bB, A-\bB, \bB-A$ are all tight, and we have
$\chi^{\dt_Z(\bB)}=\chi^{\dt_Z(A-\bB)}+\chi^{\dt_Z(A\cup\bB)}-2\chi^{\dt_Z(A)}$.
Equivalently, we obtain that $A\cap B, A\cup B, A-B, B-A$ are all tight, and 
$\chi^{\dt_Z(B)}=\chi^{\dt_Z(A\cap B)}+\chi^{\dt_Z(B-A)}-2\chi^{\dt_Z(A)}$. So we can take
$\Lc=\{A,A\cap B,B-A\}$.

\item Suppose $A\cap B, A-B\in\Sc$. So $A\cup B\notin\Sc$, $B-A\notin\Sc$.
This is symmetric to case 3), reversing the roles of $A$ and $B$. 
We obtain that $A\cap B, A\cup B, A-B, B-A$ are all tight,
and $\chi^{\dt_Z(A)}=\chi^{\dt_Z(A\cap B)}+\chi^{\dt_Z(A-B)}-2\chi^{\dt_Z(B)}$. So we can
take $\Lc=\{A,A\cap B,A-B\}$. \qedhere
\end{enumerate} 
\end{proof}

\begin{lemma} \label{laminar}
Let $f:2^V\mapsto\Z$ be a symmetric, normalized, two-way uncrossable, even-parity function.
Let $\Sc:=\{S\sse V: f(S)\geq 3\}$. Let $E'\sse E$. 
Let $\hx\in\R^{E'}$ be an extreme-point solution to the following system:
\begin{equation}
x\bigl(\dt_{E'}(S)\bigr)\geq f(S) \quad \forall S\in\Sc, \qquad
0\leq x_e\leq 1 \quad \forall e\in E'.
\label{fndlp}
\end{equation}
Let $F=\{e\in E':0<\hx_e<1\}$. There exists a laminar family $\Lc\sse\Sc$ with $|\Lc|=|F|$
such that the vectors $\bigl\{\chi^{\dt_F(S)}\bigr\}_{S\in\Lc}$ are linearly independent, and
$\hx\bigl(\dt(S)\bigr)=f(S)$ for all $S\in\Lc$.
\end{lemma}

\begin{proof}
Let $Z=\supp(\hx)$. 
We say that $S\sse V$ is tight if $\hx\bigl(\dt(S)\bigr)=f(S)$.
Consider a maximal laminar family $\T$ of tight sets from $\Sc$ such that the vectors
$\bigl\{\chi^{\dt_Z(S)}\bigr\}_{S\in\T}$ are linearly independent. 
We first claim that if $T\in\Sc$
is tight, then $\chi^{\dt_Z(T)}\in\spn\bigl(\bigl\{\chi^{\dt_Z(S)}\bigr\}_{S\in\T}\bigr)$. 
Suppose not. Then, among all tight sets $T\in\Sc$
such that $\chi^{\dt_Z(T)}\notin\spn\bigl(\bigl\{\chi^{\dt_Z(S)}\bigr\}_{S\in\T}\bigr)$, let
$B\in\Sc$ be a tight set that weakly crosses the fewest number of sets in $\T$.
Let $A\in\T$ be a set that weakly crosses $B$. By Lemma~\ref{lcspan}, there is a laminar
family $\Lc'\sse\Sc\cap\{A,A\cap B,A\cup B,A-B,B-A\}$ of tight sets such that 
$\chi^{\dt_Z(B)}\in\spn\bigl(\bigl\{\chi^{\dt_Z(S)}\bigr\}_{S\in\Lc'}\bigr)$.
All four sets $A\cap B,A\cup B,A-B,B-A$ weakly cross fewer sets in $\T$ than $B$. So by the
choice of $B$, it follows that 
$\chi^{\dt_Z(T)}\in\spn\bigl(\bigl\{\chi^{\dt_Z(S)}\bigr\}_{S\in\T}\bigr)$ for all
$T\in\Lc'$. But then we have 
$\chi^{\dt_Z(B)}\in\spn\bigl(\bigl\{\chi^{\dt_Z(S)}\bigr\}_{S\in\T}\bigr)$, which yields
a contradiction.

Consider the matrix $M$ with columns corresponding to edges in $Z$, and rows 
$\bigl\{\chi^{\dt_Z(S)}\bigr\}_{S\in\T}$,
and the submatrix $M'$ of $M$ with the same rows as $M$ but columns corresponding to
edges in $F$. 
Note that $M$ has full row rank.
We claim that $M'$ has full column rank.  
Suppose not and $M'd'=0$ for $d'\in\R^F$, $d'\neq 0$. Define $d\in\R^{E'}$ by setting
$d_e=d'_e$ if $e\in F$, and $d_e=0$ otherwise. 
Then, for a suitably small $\ve>0$, we have that $\hx\pm\ve d$ is feasible to \eqref{fndlp}, 
since $Md=0$ implies, from the previous paragraph, 
that $d\bigl(\dt(S)\bigr)=0$ for every tight set $S\in\Sc$. 
This contradicts that $\hx$ is an extreme point of \eqref{fndlp}.

So we have $|F|=\text{column rank of $M'$}=\text{row rank of $M'$}\leq |\T|$.
So there is an $|F|\times |F|$ full-rank submatrix $M''$ of $M'$. Since
$\T$ is a laminar family of tight sets, the sets of $\T$ corresponding to the rows of
$M''$ yield the desired laminar family $\Lc$.
\end{proof}

Given this laminar structure, Lemma~\ref{largereq} implies the required integrality
property. 

\begin{lemma} \label{largereq}
Let $F\sse E$, and $\Lc\sse 2^V$ be a laminar family such that $|\Lc|=|F|$ and the vectors
$\bigl\{\chi^{\dt_F(S)}\bigr\}_{S\in\Lc}$ are linearly independent. 
Let $z\in(0,1)^F$ be 
such that $z\bigl(\dt_F(S)\bigr)$ is an integer for all $S\in\Lc$. 
Then, there is some set $T\in\Lc$ such that
(a) $|\dt_F(T)|\leq 3$, and hence (b) $z\bigl(\dt_F(T)\bigr)\leq 2$.
\end{lemma}

For ease of exposition and to avoid detracting the reader, we defer the proof of
Lemma~\ref{largereq} to the end of this section, and complete first the analysis
establishing Theorem~\ref{kecss-thm} assuming this result.

\begin{proofof}{Theorem~\ref{kecss-thm}}
We argue that Algorithm~\ref{kecss-alg} terminates in at most $|E|$ iterations.
The function $f^\res$ is two-way uncrossable and has even parity (and is symmetric,
normalized), due to Lemma~\ref{twoway} and Claim~\ref{even}. 
Let $F=\{e\in E': 0<\hx_e<1\}$, and $\Lc\sse\Sc$ be the laminar family
corresponding to $\hx$ given by Lemma~\ref{laminar}.  
Let $Q=\{e\in E': \hx_e=1\}$. Let $z=(\hx_e)_{e\in F}$. 
Then, $z\bigl(\dt_F(S)\bigr)=f^\res(S)-|\dt_Q(S)|$ is an integer for all $S\in\Lc$, where
$f^\res$ is the function at the start of the iteration.
So $F$, $\Lc$, and $z$ 
satisfy the conditions of Lemma~\ref{largereq}. So we have 
$z\bigl(\dt_F(T)\bigr)=f^\res(T)-|\dt_Q(T)|\leq 2$ for some $T\in\Lc$. 
But $T\in\Sc$, so $f^\res(T)\geq 3$. This implies that
$\dt_Q(T)\neq\es$, so $Q\neq\es$, and after the update in step~\ref{kecss-update}, $T$
drops out of $\Sc$. So in each iteration, 
we add at least one new edge to $H$, and hence will terminate in at most $|E|$ iterations.

By definition, at termination, we have that $H$ is at least $(k-2)$-edge connected. To
bound $c(H)$, we claim that we have $c(H)+\lpopt[{\text{\ref{reslp}}}]\leq\lpopt$ at the
start and end of every iteration. This clearly holds at the start of the first iteration, since we
have $H=\es$, and \eqref{reslp} is the same as \eqref{kecss-lp} then. In a given
iteration, if we add $Q\sse E'$ to $H$, then $(\hx_e)_{e\in E'-Q}$ is a feasible solution
to the residual LP at the start of the next iteration. So the optimal value of the
residual LP decreases by at least $c(Q)$, and $c(H)$ increases by $c(Q)$. Hence, if the
bound holds at the start of an iteration, it also holds at the start of the next
iteration. At the end of the final iteration, $\Sc=\es$, so the final solution $H$
satisfies $c(H)\leq\lpopt$.
\end{proofof}

\begin{proofof}{Lemma~\ref{largereq}}
Part (b) follows immediately from part (a), because $z\bigl(\dt_F(T)\bigr)<|\dt_F(T)|\leq 3$, since 
$z_e<1$ for all $e\in F$, and is an integer.

Part (a) follows from a token counting argument in~\cite{Jai01} (see Theorem 4.4
in~\cite{Jai01}), which we include below for completeness. 
Suppose that no such set $T$ exists. We create two tokens for each edge $uv\in F$, giving
one token to each endpoint. 
We view $\Lc$ as a forest in the usual
fashion, where the children of a set $S\in\Lc$ are the maximal sets in $\Lc$ strictly
contained in $S$. 
We assign the tokens at each node $u$ to the minimal set in $\Lc$ containing $u$. 
Thus, every leaf set in $\Lc$ receives at least $4$ tokens. Also, if $S\in\Lc$ has only
one child $R$, then we claim that $S$ receives at least $2$ tokens. 
This is because $z\bigl(\dt(S)\bigr)-z\bigl(\dt(R)\bigr)$ is a non-zero integer,
since $\chi^{\dt_F(S)}$ and $\chi^{\dt_F(R)}$ are linearly independent,
and $0<z_e<1$ for all $e\in F$, so it must be that there are at least two edges in $F$
having an endpoint in $S-R$. 

Now by induction on $|\Lc|$ it is easy to show that 
the number of tokens received by sets in $\Lc$ is at least $2|\Lc|+2$. This yields
a contradiction since the number of tokens created is exactly $2|F|=2|\Lc|$.
The base case when $|\Lc|=1$ follows
trivially. For the induction step, consider a maximal set $S\in\Lc$. Let $n_S$ be
the number of sets in $\Lc$ contained in $S$ (including $S$), and let $a_S$ be the number
of children of $S$. If $a_S=0$, then $n_S=1$, and $S$, being a leaf, receives at least
$2n_S+2$ tokens. Suppose $a_S\geq 1$.
Applying the induction hypothesis to each child of $S$, and summing,
the total number of tokens received by the sets in $\Lc$ strictly contained in $S$ is at
least $2(n_S-1)+2a_S=2n_S+2(a_S-1)$.
If $a_S=1$, then $S$ receives $2$ tokens, and so together with the above number of tokens,
this yields $2n_S+2$ tokens received by the sets in $\Lc$ contained in $S$ (including
$S$). 
If $a_S\geq 2$, then the above number is already at least $2n_s+2$.
Summing over all maximal sets in $\Lc$ shows that the number of tokens received by sets in
$\Lc$ is at least $2|\Lc|+2$.
\end{proofof}

\begin{remark} \label{algimprov}
Lemmas~\ref{lcspan} and~\ref{laminar} also hold when we expand $\Sc$ to include all sets
with $f(S)\geq 2$, and in fact, in this case, we do not need $f$ to have even parity (see
Lemma~\ref{new-lcspan}, Corollary~\ref{new-laminar}).  
However, Lemma~\ref{largereq} {\em is tight}, i.e., the bounds in parts (a) and (b) are
tight and the lemma does not hold if we decrease either of these bounds. In
Appendix~\ref{largereq-tight}, we give an example where every set $S\in\Lc$ has
$|\dt_F(S)|\geq 3$ and $z\bigl(\dt(S)\bigr)$ is an integer that is at least $2$.
This is the chief barrier in extending our algorithm
and analysis to obtain $(k-1)$ edge-connectivity.
\end{remark}

\section{\boldmath $(k-1)$-edge connectivity incurring cost at most $1.5\cdot\lpopt$}
\label{new-kecss} \label{bicriteria}
We now show how to modify Algorithm~\ref{kecss-alg} and its analysis to obtain a
$(k-1)$-edge connected subgraph of cost at most $1.5\cdot\lpopt$ {\em for all $k$}
(Theorem~\ref{bicriteria-kecss}). The 
modification to the algorithm is simple: we now drop sets when their residual requirement
becomes at most $1$ (as opposed to $2$). This yields $(k-1)$-edge connectivity. 
As stated in Remark~\ref{algimprov},
Lemmas~\ref{lcspan} and~\ref{laminar} continue to hold with this change, and without
requiring even parity; this is the key thing we need to prove in the analysis, and is
shown in Lemma~\ref{new-lcspan} and Corollary~\ref{new-laminar}.  
Lemma~\ref{largereq} no longer yields that we have an integral edge; but, 
{\em it does yield that there is some edge with $\hx_e\geq\frac{2}{3}$}, where $\hx$ is
the extreme-point solution to the current LP. So we can pick this edge and continue, and
this leads to cost at most $1.5\cdot\lpopt$.
To obtain the guarantee with unit edge costs, it will be convenient to also delete edges
not in $\supp(\hx)$.

\begin{procedure}[ht!] 
\caption{$1.5$Apx-$k$-ECSS() 
\textnormal{\qquad // $1.5$-approximation LP-rounding algorithm for \kecss}
\label{bi-kecssalg}} 
\KwIn{\kecss instance $\bigl(G=(V,E),\{c_e\geq 0\}_{e\in E},\,k\geq 2\bigr)$}
\KwOut{subset $H\sse E$}

Initialize $E'\assign E$, $H\assign\es$, $f^\res=f^{\kecss}$. 

Let $\Sc=\{S\sse V: f^\res(S)\geq 2\}$.

\While{$H$ is not $(k-1)$-edge connected (equivalently $\Sc\neq\es$)}{
Compute an extreme-point optimal solution $\hx\in\R^{E'}$ to the following residual LP:
\begin{alignat*}{3}
\min & \quad & \sum_{e\in E'}c_ex_e \tag{\ref{reslp}} \\
\text{s.t.} \quad 
&& x\bigl(\dt_{E'}(S)\bigr) & \geq f^\res(S) \qquad && \forall S\in\Sc \\
&& 0 \leq x_e & \leq 1 && \forall e\in E'. \notag
\end{alignat*}

Update $H\assign H\cup\bigl\{e\in E':\hx_e\geq\frac{2}{3}\bigr\}$, 
$E'\assign\bigl\{e\in E': 0<\hx_e<\frac{2}{3}\bigr\}$,
$f^\res(S)=f^{\kecss}(S)-|\dt_H(S)|$ for all
$S\sse V$, and $\Sc=\{S\sse V: f^\res(S)\geq 2\}$. \label{bi-kecssupdate}
}

\Return $H$.
\end{procedure}

\begin{lemma}[{\bf\boldmath Uncrossing tight sets with requirement at least $2$}]
\label{new-lspan} \label{new-lcspan}
Let $f:2^V\mapsto\Z$ be a symmetric, normalized, two-way uncrossable function.
Let $\Sc:=\{S\sse V: f(S)\geq 2\}$.
Let $x\in\R_+^E$ satisfy $x\bigl(\dt(S)\bigr)\geq f(S)$ for all $S\in\Sc$.
Let $Z=supp(x)$. 
Say that a set $S\sse V$ is tight if $x\bigl(\dt(S)\bigr)=f(S)$.
Let $A,B\in\Sc$ be two weakly-crossing tight sets. 
Then, there exists a laminar family $\Lc\sse\Sc\cap\{A,A\cap B,A\cup B,A-B,B-A\}$ of tight
sets with $A\in\Lc$ such that 
$\chi^{\dt_Z(B)}\in\spn\bigl(\bigl\{\chi^{\dt_Z(S)}\bigr\}_{S\in\Lc}\bigr)$.
\end{lemma} 

\begin{proof}
The proof is quite similar to that of Lemma~\ref{lcspan}, and we briefly sketch the
changes. 
Let $Z'=\{uv\in Z: u\in A-B, v\in B-A\}$ and $\tht=x(Z')$,
and $Z''=\{uv\in Z: u\in A\cap B, v\notin A\cup B\}$ and $\gm=x(Z'')$
(see Fig.~\ref{tightuncross}).
As before, since $f$ is two-way uncrossable, at least one of $A\cap B, A\cup B$ is in
$\Sc$, and at least one of $A-B, B-A$ is in $\Sc$. 
Also, if $x\bigl(\dt(S)\bigr)\geq f(S)$ for $S\in\{A\cap B, A\cup B\}$, then both sets are
tight, and $\tht=0$; 
if $x\bigl(\dt(S)\bigr)\geq f(S)$ for $S\in\{A-B, B-A\}$, then both sets are tight, and
$\gm=0$. 
So if $\{A\cap B,A\cup B\}\sse\Sc$ or $\{A-B,B-A\}\sse\Sc$, then we are done.
  
We only need to change slightly the portion of the argument when exactly one of 
$A\cap B, A\cup B$ lies in $\Sc$, and exactly one of $A-B,B-A$ lies in $\Sc$. We consider
only the case where $A\cup B, A-B\in\Sc$, and so $A\cap B\notin\Sc$, $B-A\notin\Sc$. The
other cases follow from symmetric arguments exactly as in Lemma~\ref{lcspan}.

Let $Z_{A\cap B,B-A}=\{uv\in Z: u\in A\cap B, v\in B-A\}$ and $\al=x(Z_{A\cap B,B-A})$
(see Fig.~\ref{tightuncross}).
Again, we first show that 
$x\bigl(\dt(A\cap B)\bigr)+x\bigl(\dt(B-A)\bigr)\geq f(A\cap B)+f(B-A)$. If this does not
hold, then we obtain, as before the inequality
\begin{equation*}
f(B)=x\bigl(\dt(B)\bigr)=x\bigl(\dt(A\cap B)\bigr)+x\bigl(\dt(B-A)\bigr)-2\al
<f(A\cap B)+f(B-A)-2\al. \tag{\ref{lcineq2}}
\end{equation*}
But now since $f(B)\geq 2$ and $f(A\cap B), f(B-A)\leq 1$, this immediately yields a
contradiction. So since 
$x\bigl(\dt(A\cap B)\bigr)+x\bigl(\dt(B-A)\bigr)\geq f(A\cap B)+f(B-A)$, the same
arguments as in Lemma~\ref{lcspan}, establish that we must have 
$x\bigl(\dt(S)\bigr)\geq f(S)$ for $S\in\{A\cap B,B-A\}$, and so all four sets 
$A\cap B, A\cup B, A-B, B-A$ are tight, and $\tht=\gm=0$. Given this, we also obtain that
$\al=0$: this is because, similar to \eqref{lcineq2}, we have 
$f(B)=f(A\cap B)+f(B-A)-2\al$, and $f(B)\geq 2$, $f(A\cap B), f(B-A)\leq 1$.
Therefore, in this case, we have 
$\chi^{\dt_Z(B)}=\chi^{\dt_Z(A-B)}+\chi^{\dt_Z(A\cup B)}-2\chi^{\dt_Z(A)}$.
\end{proof}

Lemma~\ref{new-lcspan} immediately yields the following laminarity characterization of
extreme points. The proof is the same as that of Lemma~\ref{laminar}, and is
therefore omitted. 

\begin{corollary} \label{new-laminar}
Let $f:2^V\mapsto\Z$ be a symmetric, normalized, two-way uncrossable function.
Let $\Sc:=\{S\sse V: f(S)\geq 2\}$. Let $E'\sse E$. 
Let $\hx\in\R^{E'}$ be an extreme-point solution to the following system:
\begin{equation*}
x\bigl(\dt_{E'}(S)\bigr)\geq f(S) \quad \forall S\in\Sc, \qquad
0\leq x_e\leq 1 \quad \forall e\in E'.
\end{equation*}
Let $F=\{e\in E':0<\hx_e<1\}$. There exists a laminar family $\Lc\sse\Sc$ with $|\Lc|=|F|$
such that the vectors $\bigl\{\chi^{\dt_F(S)}\bigr\}_{S\in\Lc}$ are linearly independent, and
$\hx\bigl(\dt(S)\bigr)=f(S)$ for all $S\in\Lc$.
\end{corollary}

\begin{proofof}{Theorem~\ref{bicriteria-kecss}}
We can solve \eqref{reslp} efficiently, because a polytime separation oracle follows as in
Lemma~\ref{lpsolve}. 
We argue that Algorithm~\ref{bi-kecssalg} terminates in at most $|E|$ iterations.
The function $f^\res$ is two-way uncrossable and has even parity (and is symmetric,
normalized), due to Lemma~\ref{twoway} and Claim~\ref{even}. 
Let $F=\{e\in E': 0<\hx_e<1\}$, and $\Lc\sse\Sc$ be the laminar family
corresponding to $\hx$ given by Corollary~\ref{new-laminar}.  
Let $Q=\{e\in E': \hx_e=1\}$. Let $z=(\hx_e)_{e\in F}$. 
Then, $z\bigl(\dt_F(S)\bigr)=f^\res(S)-|\dt_Q(S)|$ is an integer for all $S\in\Lc$, where
$f^\res$ is the function at the start of the iteration.
So $F$, $\Lc$, and $z$ 
satisfy the conditions of Lemma~\ref{largereq}. So we have 
$z\bigl(\dt_F(T)\bigr)=f^\res(T)-|\dt_Q(T)|\leq 2$, and $|\dt_F(T)|\leq 3$ for some
$T\in\Lc$.  
If $z\bigl(\dt_F(T)\bigr)\leq 1$, then since $f^\res(T)\geq 2$ (as $T\in\Sc$), this
implies that $\dt_Q(T)\neq\es$, so $Q\neq\es$; otherwise, there is some $e\in\dt_F(T)$
with $z_e=\hx_e\geq\frac{2}{3}$. 
In either case, after the update in step~\ref{bi-kecssupdate}, at least one edge gets
added to $H$, and $T$ drops out of $\Sc$. 
Hence, we terminate in at most $|E|$ iterations.

By definition, at termination, we have that $H$ is at least $(k-1)$-edge connected. 
To bound the cost, we claim that 
$c(H)+1.5\cdot\lpopt[{\text{\ref{reslp}}}]\leq 1.5\cdot\lpopt$ at the
start and end of every iteration, which, in particular, implies that the final solution $H$
satisfies $c(H)\leq1.5\cdot\lpopt$. The inequality clearly holds at the beginning.
In a given iteration, if we add $E''\sse E'$ to $H$, then $(\hx_e)_{e\in E'-E''}$ is a
feasible solution to the residual LP at the start of the next iteration. So the optimal
value of the residual LP decreases by at least $\frac{2}{3}\cdot c(E'')$, and $c(H)$
increases by $c(E'')$. So we maintain the invariant. 

\vspace*{-1ex}
\paragraph{Unit edge costs.} With unit edge costs, we can refine the above cost
analysis, by arguing in a manner similar to~\cite{GabowGTW09}. 
Note that $\lpopt\geq\frac{kn}{2}$ in this case, since $x\bigl(\dt(v)\bigr)\geq k$ for
every vertex $v$.  
Let $F_1=\{e\in E': 0<\hx_e<1\}$, and $\Lc_1\sse\Sc$ be the laminar family
corresponding to $\hx$ given by Corollary~\ref{new-laminar} in the very first iteration,
that is, when we solve \eqref{kecss-lp}. We have $|F_1|=|\Lc_1|\leq 2n$ since $\Lc_1$ is a 
laminar collection of subsets of a ground set with $n=|V|$ elements.
We argue that $c(H)=|H|\leq\lpopt+\min\{|F_1|-\hx(F_1),\hx(F_1)/2\}$.
Any increase in cost beyond $\lpopt$ results from rounding fractional edges to $1$, and
the increase when edge $e\in F_1$ gets rounded to $1$ is $1-\hx_e$.
So we have $c(H)\leq\lpopt+|F_1|-\hx(F_1)$. 
Let $R\sse F_1$ be the set of edges from $F_1$ that got added to $H$ in the first
iteration, so $|R|\leq 1.5\cdot\hx(R)$.
Then, $(\hx_e)_{e\in F_1-R}$ is a solution to the residual LP at the next iteration, and
so by the analysis above, the edges added to $H$ after the first iteration have cost at
most $1.5\cdot\hx(F_1-R)$. So we have 
$c(H)\leq\lpopt-x(F_1)+|R|+1.5\cdot\hx(F_1-R)\leq\lpopt+\hx(F_1)/2$. 
Thus, 
\begin{equation*}
\begin{split}
c(H)-\lpopt & \leq\min\Bigl\{|F_1|-\hx(F_1),\tfrac{\hx(F_1)}{2}\Bigr\}\leq
\tfrac{1}{3}\cdot\bigl(|F_1|-\hx(F_1)\bigr)+\tfrac{2}{3}\cdot\tfrac{\hx(F_1)}{2}
\leq\frac{|F_1|}{3}\leq\frac{2n}{3} \\
& \leq\frac{4}{3k}\cdot\lpopt. \qedhere
\end{split}
\end{equation*}
\end{proofof}

\section{\boldmath $k$-edge connected spanning multigraph} \label{kecsm}
We now consider the $k$-edge connected spanning multigraph (\kecsm) problem, which is the
variant of \kecss, where we allowed to pick multiple copies of an edge. 
We correspondingly modify \eqref{kecss-lp} to obtain the following LP-relaxation for
\kecsm:
\begin{alignat}{3}
\min & \quad & \sum_{e\in E}c_ex_e \tag{\kecsmlp} \label{kecsm-lp} \\
\text{s.t.} \quad && x\bigl(\dt(S)\bigr) & \geq k
\qquad && \forall \es\neq S\subsetneq V \label{ecsmcon} \\
&& x & \geq 0. \notag
\end{alignat}
In Section~\ref{contrib}, we sketched a proof showing that our guarantee for \kecss
(Theorem~\ref{kecss-thm}) immediately yields a pseudopolynomial time algorithm yielding
the guarantee in Theorem~\ref{kecsm-thm}. We show here that Algorithm~\ref{kecss-alg} can
be adapted to yield a polytime algorithm for \kecsm achieving the guarantee stated in
Theorem~\ref{kecsm-thm}, thereby proving Theorem~\ref{kecsm-thm}. 

In Algorithm~\ref{kecsm-alg}, we describe how to obtain, for even $k\geq 4$, a multigraph
$H$ that is $(k-2)$-edge connected, with $c(H)\leq\lpopt[{\kecsmlp}]$. 
The idea is to simply pick the ``integral portion'' of an extreme-point optimal solution
to \eqref{kecsm-lp} initially, which leaves one with a residual problem and an LP solution
to this residual problem that picks each edge to an extent of at most $1$.
To elaborate, 
we compute an extreme-point optimal solution
$\hx$ to \eqref{kecsm-lp}. 
Clearly, we have $|\supp(\hx)|\leq |E|$.
We pick $\floor{\hx_e}$ copies of every edge $e$,
and move to $f^\res$ given by 
$f^\res(S)=f^{\kecss}(S)-\sum_{e\in\dt_G(S)}\floor{\hx_e}$ for all $S\sse V$. 
Now $\bigl(\hx_e-\floor{\hx_e}\bigr)_e$
yields a feasible solution to \eqref{reslp} (where we still have the $0\leq x_e\leq 1$
constraints), and we proceed thereafter as in Algorithm~\ref{kecss-alg}. The analysis in
Section~\ref{kecss} continues to hold, since $f^\res$ remains a two-way uncrossable,
even-parity function, and so we obtain a multigraph that is $(k-2)$-edge-connected after
at most $|E|$ iterations.   

\begin{procedure}[ht!] 
\caption{$k$-ECSM-Alg() \textnormal{\qquad // LP-rounding algorithm for \kecsm}
\label{kecsm-alg}} 
\KwIn{\kecsm instance $\bigl(G=(V,E),\{c_e\geq 0\}_{e\in E},\,\text{even}\ k\geq 4\bigr)$}
\KwOut{Capacity-vector $w\in\Z_+^E$ specifying edge multiplicities; equivalently, a
  multigraph $H$}

Compute an extreme-point optimal solution $\hx$ to \eqref{kecsm-lp}.

Initialize $w_e=\floor{\hx_e}$ for all $e\in E$, $E'\assign\{e\in E:\hx_e>w_e\}$, 
and $f^\res(S)=f^{\kecss}(S)-w\bigl(\dt_G(S)\bigr)$ for all $S\sse V$.

Let $\Sc=\{S\sse V: f^\res(S)\geq 3\}$.

\While{$G$ is not $(k-2)$-edge connected under capacity-vector $w$ (equivalently $\Sc\neq\es$)}{
Compute an extreme-point optimal solution $\hx\in\R^{E'}$ to \eqref{reslp}.

Update $w_e\assign w_e+\floor{\hx_e}$ for all $e\in E'$,
$E'\assign\{e\in E': \hx_e>\floor{\hx_e}\}$,  
$f^\res(S)=f^{\kecss}(S)-w\bigl(\dt_G(S)\bigr)$ for all
$S\sse V$, and $\Sc=\{S\sse V: f^\res(S)\geq 3\}$. \label{kecsm-update}
}

\Return $w$.
\end{procedure}

\begin{proofof}{Theorem~\ref{kecsm-thm}}
To obtain the guarantee in Theorem~\ref{kecsm-thm}, we simply run
Algorithm~\ref{kecsm-alg} with parameter $k+2$ if $k$ is even, and parameter $k+3$ if $k$
is odd. Note that for any $k'\geq k$, we have
$\lpopt[{\kecsmlp[{k'}]}]\leq\frac{k'}{k}\cdot\lpopt[{\kecsmlp}]$.

So it suffices to show that for even $k\geq 4$, Algorithm~\ref{kecsm-alg} returns a
multigraph $H$ that is at least $(k-2)$-edge connected and satisfies
$c(H)\leq\lpopt[{\kecsmlp}]$. 
This follows from the analysis in Section~\ref{kecss}. 
As in lemma~\ref{lpsolve}, we obtain a polytime separation oracle for \eqref{reslp} by
computing the min-cut value, and considering all $2$-approximate min-cuts, under the
capacity vector $w+\hx$. Moreover, Lemmas~\ref{lcspan}--\ref{largereq} continue to hold,   
since $f^\res$ is two-way uncrossable and has even parity.
\end{proofof}

\section{Degree-bounded \boldmath \kecss and \kecsm} \label{degbnd-kecss}
We now consider the degree-bounded versions of \kecss and \kecsm, and showcase the
versatility of our simple ``uncrossing-to-achieve-laminarity'' approach, to
obtain strong guarantees for these problems.
In the {\em minimum-degree \kecss} (\mdkecss) problem, we are given a \kecss instance
$\bigl(G=(V,E),\{c_e\}_{e\in E},k\bigr)$ along with (integer) degree lower- and upper-
bounds $\{\ldeg_v,\udeg_v\}_{v\in V}$ on the nodes, and the goal is to find a minimum-cost
$k$-edge connected subgraph $H$ satisfying the degree bounds $\ldeg_v\leq|\dt_H(v)\leq\udeg_v$
for all $v\in V$. Minimum-degree \kecsm (\mdkecsm) is defined similarly, except that $H$ is
allowed to be a multigraph (as in \kecsm). 

We first discuss minimum-degree \kecss, and then show how to handle
\mdkecsm by suitably tweaking the underlying algorithm. 
The LP-relaxation \eqref{kecss-lp} for \kecss extends naturally to yield the following
natural LP-relaxation for \mdkecss.
\begin{alignat}{3}
\min & \quad & \sum_{e\in E}c_ex_e \tag{\mdkecsslp} \label{mdkecss-lp} \\
\text{s.t.} \quad && x\bigl(\dt(S)\bigr) & \geq k
\qquad && \forall \es\neq S\subsetneq V \label{mdcon} \\
&& 0 \leq x_e & \leq 1 && \forall e\in E \label{mdbnd} \\
&& \ldeg_v \leq x\bigl(\dt(v)\bigr) & \leq \udeg_v && \forall v\in V. \label{degbnd}
\end{alignat}
We collectively refer to constraints \eqref{degbnd} as degree constraints.
We show that Algorithm~\ref{kecss-alg} readily extends to yield the following guarantee.

\begin{theorem} \label{mdkecss-thm}
There is a polytime algorithm that, for any even $k$, returns a $(k-2)$-edge connected
subgraph $H$ such that $c(H)\leq\lpopt[{\mdkecsslp}]$ and
$\ldeg_v-2\leq|\dt_H(v)|\leq\udeg_v+2$ for all $v\in V$.

For odd $k$, we obtain that $H$ is $(k-3)$-edge connected, with the same bounds on $c(H)$
and $\bigl\{|\dt_H(v)|\bigr\}_{v\in V}$.
\end{theorem}

As discussed in Section~\ref{intro}, prior work on degree-bounded network design, when
specialized to degree-bounded \kecss, yields results where the subgraph $H$
returned is $k$-edge connected, but $c(H)$ is roughly $2\cdot\lpopt[{\mdkecsslp}]$, and
there is some violation of the degree constraints, where the degree violation involves
a multiplicative $O(1)$ factor or an additive $O(k)$ factor.

\subsubsection*{Proof of Theorem~\ref{mdkecss-thm}}
We may assume that $k$ is even, as the guarantee for odd $k$ is obtained by
invoking the algorithm with $k-1$. We mimic Algorithm~\ref{kecss-alg}. 
Essentially everything works as before, since tight degree constraints correspond to
singleton nodes, and these can be added to any laminar family while preserving
laminarity. Correspondingly, dropping a set may involve dropping a connectivity
constraint \eqref{mdcon} and/or dropping a degree constraint \eqref{degbnd}.
The residual LP is now augmented with node degree constraints (see \eqref{mdreslp}) for
the nodes $V'\sse V$ whose degree constraints have not been dropped.
We compute an extreme-point solution $\hx$ to this LP. 
We now both pick edges whose $\hx_e$s are set to $1$, and also delete edges
whose $\hx_e$s are set to $0$, and when we pick edges, we update both the residual
requirements of sets, and the degree lower and upper bounds for nodes.
Our laminar family may now also involve some singleton nodes whose degree constraints,
either the lower- or upper- bound in \eqref{degbnd}, are tight (see
Lemma~\ref{mdlaminar}). Given this laminar 
structure, we can again find a tight set $S$, which could be a singleton, that has at most
$3$ fractional edges crossing it. We drop the constraint(s) corresponding to this set; if
this set is a singleton node $\{v\}$, then this means also dropping both the lower- and
upper- bound degree constraints for $v$. 
In the latter case, we have
$\hx\bigl(\dt(v)\bigr)\in\{1,2\}$,
so the additive violation in the degree lower and upper bound for $v$ is at most $2$. 
The complete description appears as Algorithm~\ref{mdkecss-alg} below.

\begin{procedure}[t!] 
\caption{MD-$k$ECSS-Alg() \textnormal{\qquad // LP-rounding algorithm for \mdkecss}
\label{mdkecss-alg}} 
\KwIn{\kecss instance $\bigl(G=(V,E),\{c_e\geq 0\}_{e\in E},\,\text{even}\ k\bigr)$,
integers $\ldeg_v\leq\udeg_v$ for all $v\in V$} 
\KwOut{subset $H\sse E$}

Initialize $E'\assign E$, $V'\assign V$, $H\assign\es$, $f^\res=f^{\kecss}$. 
Set $\ldeg^\res_v=\ldeg_v, \udeg^\res_v=\udeg_v$ for all $v\in V$. 

Let $\Sc=\{S\sse V: f^\res(S)\geq 3\}$.

\While{$H$ is not $(k-2)$-edge connected (equivalently $\Sc\neq\es$) or $V'\neq\es$}{
Compute an extreme-point optimal solution $\hx\in\R^{E'}$ to the following residual LP:
\begin{alignat}{3}
\min & \quad & \sum_{e\in E'}c_ex_e \tag{MD-ResLP} \label{mdreslp} \\
\text{s.t.} \quad 
&& x\bigl(\dt_{E'}(S)\bigr) & \geq f^\res(S) \qquad && \forall S\in\Sc \label{mdrescon} \\
&& 0 \leq x_e & \leq 1 && \forall e\in E' \notag \\
&& \ldeg^\res_v \leq x\bigl(\dt(v)\bigr) & \leq \udeg^\res_v \qquad && \forall v\in V'.
\end{alignat}

Update $H\assign H\cup\{e\in E': \hx_e=1\}$, $E'\assign\{e\in E': 0<\hx_e<1\}$,
$f^\res(S)=f^{\kecss}(S)-|\dt_H(S)|$ for all $S\sse V$, and
$\Sc=\{S\sse V: f^\res(S)\geq 3\}$. \label{mdkecss-updatesc} 

Update $\ldeg^\res_v=\ldeg_v-|\dt_H(v)|$ and $\udeg^\res_v=\udeg_v-|\dt_H(v)|$, for all
$v\in V$. 

Update $V'=\{v\in V': \hx\bigl(\dt_{E'}(v)\bigr)>2\text{ or }|\dt_{E'}(v)|-\hx\bigl(\dt_{E'}(v)\bigr)>2\}$.
\label{mdkecss-update}
}

\Return $H$.
\end{procedure}

The analysis is very similar to that of Algorithm~\ref{kecss-alg}. A separation oracle for
\eqref{mdreslp} follows exactly as before, so \eqref{mdreslp} can be solved
efficiently. (Note that when $k=2$, there are no constraints \eqref{mdrescon}.)
Lemma~\ref{laminar} essentially continues to hold, but needs to be restated
more precisely, to take into account tight degree constraints.

\begin{lemma} \label{mdlaminar}
Let $f:2^V\mapsto\Z$ be a symmetric, normalized, two-way uncrossable, even-parity
function.
Let $\lb_v,\ub_v\in\Z_+$ for all $v\in V$.
Let $\Sc:=\{S\sse V: f(S)\geq 3\}$. Let $E'\sse E$, and $V'\sse V$. 
Let $\hx\in\R^{E'}$ be an extreme-point solution to the following system:
\begin{equation*}
x\bigl(\dt_{E'}(S)\bigr)\geq f(S) \quad \forall S\in\Sc, \qquad
0\leq x_e\leq 1 \quad \forall e\in E', \qquad
\lb_v \leq x\bigl(\dt(v)\bigr) \leq \ub_v \quad \forall v\in V'.
\end{equation*}
Let $F=\{e\in E':0<\hx_e<1\}$. There exists a laminar family $\Lc\sse\Sc$ and $W\sse V'$
with $|\Lc|+|W|=|F|$ such that
\begin{enumerate*}[label=\textnormal{(\alph*)}, topsep=0.1ex, noitemsep, leftmargin=*] 
\item the vectors $\bigl\{\chi^{\dt_F(S)}\bigr\}_{S\in\Lc},\,\bigl\{\chi^{\dt_F(v)}\bigr\}_{v\in W}$ are
linearly independent; 
\item $\hx\bigl(\dt(S)\bigr)=f(S)$ for all $S\in\Lc$; and 
\item $\hx\bigl(\dt(v)\bigr)\in\bigl\{\lb_v,\ub_v\}$, $\hx\bigl(\dt(v)\bigr)\geq 1$, for all
  $v\in W$.
\end{enumerate*}
\end{lemma}

We defer the proof of Lemma~\ref{mdlaminar}, and first use this to argue that
Algorithm~\ref{mdkecss-alg} terminates in at most $|E|+|V|$ iterations with the desired 
subgraph $H$.
We argue that in every iteration, either we add an edge to $H$ and drop a set from $\Sc$,
or we delete a vertex from $V'$.

Let $F=\{e\in E': 0<\hx_e<1\}$, where $E'$ refers to the edge-set at the start of the
iteration. (Note that $E'$ gets updated to $F$ at the end of the iteration.) 
Let $f^\res$ be the cut-requirement function at the start of the iteration. As before,
$f^\res$ is two-way uncrossable, symmetric, normalized, and has even parity.
Let $\Lc\sse\Sc$ and $W\sse V'$ be the laminar family and node-set
corresponding to $\hx$ given by Lemma~\ref{mdlaminar}.   
Let $Q=\{e\in E': \hx_e=1\}$. Let $z=(\hx_e)_{e\in F}$. 
Then, $z\bigl(\dt_F(S)\bigr)=f^\res(S)-|\dt_Q(S)|$ is an integer for all $S\in\Lc$, and
$z\bigl(\dt_F(v)\bigr)=\hx\bigl(\dt_{E'}(v)\bigr)-|\dt_Q(v)|$ is a positive integer for
all $v\in W$, due to Lemma~\ref{mdlaminar} (c).

So $F$, $\Lc\cup\bigl\{\{v\}: v\in W\bigr\}$, and $z$ 
satisfy the conditions of Lemma~\ref{largereq}. So either for some $T\in\Lc$,  
we have $z\bigl(\dt_F(T)\bigr)=f^\res(T)-|\dt_Q(T)|\leq 2$, or for some $v\in W$, we have
$|\dt_F(v)|\leq 3$ and $z\bigl(\dt_F(v)\bigr)\leq 2$.
In the former case, since $T\in\Sc$, we have $f^\res(T)\geq 3$, which implies that 
$\dt_Q(T)\neq\es$, and so $Q\neq\es$ and after the update in step~\ref{mdkecss-updatesc}, $T$
drops out of $\Sc$. 
In the latter case, at the end of the iteration, when $E'$ gets updated to
$E^{'\after}=F$, we have $\hx\bigl(\dt_{E^{'\after}}(v)\bigr)\leq 2$ and 
$|\dt_{E^{'\after}}(v)|-\hx\bigl(\dt_{E^{'\after}}(v)\bigr)\leq 3-1=2$; so $v$ drops out
of $V'$ in step~\ref{mdkecss-update}.

By definition, at termination, we have that $H$ is at least $(k-2)$-edge connected. 
Consider any $v\in V$. We have $|\dt_H(v)|\geq\ldeg_v-2$ because as long as $v\in V'$, we
maintain that $\hx\bigl(\dt_{E'}(v)\bigr)+|\dt_H(v)|\geq\ldeg_v$, and when $v$ drops out of
$V'$, we have $\hx\bigl(\dt_{E'}(v)\bigr)\leq 2$. Similarly, we have
$|\dt_H(v)|\leq\udeg_v+2$ because $\hx\bigl(\dt_{E'}(v)\bigr)+|\dt_H(v)|\leq\udeg_v$ as long
as $v\in V'$. Also, when $v$ drops out of $V'$, we have
$|\dt_{E'}(v)|-\hx\bigl(\dt_{E'}(v)\bigr)\leq 2$, and the degree of $v$ in the final
solution can be at most $|\dt_H(v)|+|\dt_{E'}(v)|$, which is therefore at most $\udeg_v+2$.    
The bound on $c(H)$ follows as before, given that we only include edges in $H$ when
$\hx_e=1$. 

\begin{proofof}{Lemma~\ref{mdlaminar}}
The proof is essentially the same as that of Lemma~\ref{laminar}.
Let $Z=\supp(\hx)$. 
Say that $S\sse V$ is tight if $\hx\bigl(\dt(S)\bigr)=f(S)$, and say that
$v\in V'$ is tight if $\hx\bigl(\dt(v)\bigr)\in\{\lb_v,\ub_v\}$. 
(Note that we are distinguishing above between the set $\{v\}$ being tight, and the node
$v$ being tight.)

Consider a maximal laminar family $\T$ of tight sets from $\Sc$ such that the vectors
$\bigl\{\chi^{\dt_Z(S)}\bigr\}_{S\in\T}$ are linearly independent. 
Then, exactly as argued in the proof of Lemma~\ref{laminar}, if $T\in\Sc$ is tight, we
must have $\chi^{\dt_Z(T)}\in\spn\bigl(\bigl\{\chi^{\dt_Z(S)}\bigr\}_{S\in\T}\bigr)$. 

Let $A\sse V'$ be a maximal set of tight nodes such that the vectors
$\bigl\{\chi^{\dt_Z(S)}\bigr\}_{S\in\T},\bigl\{\chi^{\dt_Z(v)}\bigr\}_{v\in A}$ are
linearly independent.
Clearly, this means that if $u\in V'$ and is tight, then 
$\chi^{\dt_Z(u)}\in\spn\bigl(\bigl\{\chi^{\dt_Z(S)}\bigr\}_{S\in\T},
\bigl\{\chi^{\dt_Z(v)}\bigr\}_{v\in A}\bigr)$.

Consider the matrix $M$ with columns corresponding to edges in $Z$, and rows 
$\bigl\{\chi^{\dt_Z(S)}\bigr\}_{S\in\T}$,
$\bigl\{\chi^{\dt_Z(v)}\bigr\}_{v\in A}$. Let $M'$ be the submatrix
of $M$ with the same rows as $M$ but columns corresponding to edges in $F$. 
The row rank of $M$ is equal to $|\T|+|A|$; so the row
rank of $M'$ is at most $|\T|+|A|$. Since $\hx$ is an extreme point and the row-span of
$M$ contains $\chi^{\dt_Z(S)}$, $\chi^{\dt_Z(u)}$ for for every tight set $S\in\Sc$ and
tight node $u\in V'$, we again have that $M'$ has full column rank.
So we have $|F|=\text{column rank of $M'$}=\text{row rank of $M'$}\leq |\T|+|A|$.
So there is an $|F|\times |F|$ full-rank submatrix $M''$ of $M'$. 
The rows of $M''$ corresponding to the sets of $\T$ yield the desired laminar family
$\Lc$, and the rows corresponding to the nodes of $A$ yield $W$.
\end{proofof}

This completes the proof of Theorem~\ref{mdkecss-thm}. \hfill \qed

\subsection{Minimum-degree \boldmath \kecsm}
We handle the minimum-degree \kecsm problem by modifying Algorithm~\ref{mdkecss-alg} in
much the same way that Algorithm~\ref{kecss-alg} was tweaked to handle \kecsm.

The LP-relaxation for \mdkecsm is almost the same as \eqref{mdkecss-lp}, except that we
replace constraints \eqref{degbnd} simply with the nonnegativity constraints $x\geq 0$. Let
(\mdkecsmlp) denote this LP-relaxation.

We proceed by computing an extreme-point optimal solution to (\mdkecsmlp).
We initialize $H$ to include $\floor{\hx_e}$ copies of each edge $e$, and
define the residual cut-requirement function $f^\res$, and residual degree lower- and
upper- bounds $\ldeg^\res_v,\udeg^\res_v$ accordingly to take this into account. 
We set $E'\assign\{e\in E:\hx_e>\floor{\hx_e}\}$, $\Sc=\{S\sse V: f^\res(S)\geq 3\}$, and 
$V'=\{v\in V: \hx\bigl(\dt_{E'}(v)\bigr)>2\text{ or }|\dt_{E'}(v)|-\hx\bigl(\dt_{E'}(v)\bigr)>2\}$ 
(where the $E'$ in the definition of $V'$ refers to the updated $E'$).
From here on, we proceed as in Algorithm~\ref{mdkecss-alg}. 

This yields a multigraph $H$ satisfying the same guarantees as in
Theorem~\ref{mdkecss-thm}, but with the stronger cost bound
$c(H)\leq\lpopt[{\mdkecsmlp}]$. To obtain a $k$-edge connected multigraph, we run the
above algorithm, taking $k'=k+2$ for even $k$, and $k'=k+3$ for odd $k$. This yields the
following.

\begin{theorem} \label{mdkecsm-thm}  
There is a polynomial time algorithm for \mdkecsm that returns a $k$-edge connected
multigraph $H$ such that $c(H)\leq\rho_k\cdot\lpopt[{\mdkecsmlp}]$, and
$\ldeg_v-2\leq|\dt_H(v)|\leq\rho_k\udeg_v+2$ for all $v\in V$, 
where $\rho_k=\bigl(1+\frac{2}{k}\bigr)$, if $k$ is even, and
$\rho_k=\bigl(1+\frac{3}{k}\bigr)$ if $k$ is odd.
\end{theorem}

\bibliographystyle{abbrv}
\bibliography{kECSS}

\appendix

\section{Proofs omitted from Section~\ref{prelim}} \label{append-prelim}

\begin{proofof}{Lemma~\ref{symnormal}}
For part~\ref{symtwoway}, let $A,B\sse V$ be crossing sets. Note that $A,B\neq\es$ and $A,B\subsetneq V$.
So $g(A)=f(T_A)$ and $g(B)=f(T_B)$, where
$T_A\in\{A,V-A\}$ and $T_B\in\{B,V-B\}$. Note that $T_A$ and $T_B$ are also crossing sets.
We have 
\begin{equation}
g(A)+g(B)=f(T_A)+f(T_B)\leq\min\bigl\{f(T_A\cap T_B)+f(T_A\cup T_B),
f(T_A-T_B)+f(T_B-T_A)\bigr\} \label{symineq1}
\end{equation} 
There are now four possibilities.
\begin{enumerate}[label=$\bullet$, topsep=0.2ex, itemsep=0.1ex, leftmargin=*]
\item $T_A=A, T_B=B$: 
the RHS of \eqref{symineq1} is then at most 
$\min\bigl\{g(A\cap B)+g(A\cup B),g(A-B)+g(B-A)\bigr\}$

\item $T_A=A, T_B=V-B$: the RHS of \eqref{symineq1} is then at most 
$\min\bigl\{g(A-B)+g(B-A),g(A\cap B)+g(A\cup B)\bigr\}$, since $T_A\cup T_B=V-(B-A)$,
$T_A-T_B=A\cap B$, $T_B-T_A=V-(A\cup B)$.

\item $T_A=V-A, T_B=B$: the RHS of \eqref{symineq1} is then at most 
$\min\bigl\{g(B-A)+g(A-B),g(A\cup B)+g(A\cap B)\}$

\item $T_A=V-A, T_B=V-B$: the RHS of \eqref{symineq1} is then at most 
$\min\bigl\{g(A\cup B)+g(A\cap B),g(B-A)+g(A-B)\}$.
\end{enumerate}
So $g$ is two-way uncrossable.

\medskip
To prove part~\ref{symndp}, suppose $x\bigl(\dt(S)\bigr)\geq f(S)$ for all 
$\es\neq S\subsetneq V$. Then, for any set $S$ where $\es\neq S\subsetneq V$, we have
$x\bigl(\dt(S)\bigr)\geq f(S)$ and 
$x\bigl(\dt(S)\bigr)=x\bigl(\dt(V-S)\bigr)\geq f(V-S)$. It follows that 
$x\bigl(\dt(S)\bigr)\geq\max\{f(S),f(V-S)\}=g(S)$. Also, clearly,
$x\bigl(\dt(\es)\bigr)=0=g(\es)$ and $x\bigl(\dt(V)\bigr)=0=g(V)$.
Conversely, if $x\bigl(\dt(S)\bigr)\geq g(S)$ for all $\es\neq S\subsetneq V$, then since
$g(S)\geq f(S)$ for all $\es\neq S\subsetneq V$, we also have 
$x\bigl(\dt(S)\bigr)\geq f(S)$ for all $\es\neq S\subsetneq V$.

\medskip
Finally, consider part~\ref{symtight}. 
If $T\in\{\es,V\}$,
clearly $x\bigl(\dt(T)\bigr)=g(T)$. Otherwise, as argued above, we also have
$x\bigl(\dt(S)\bigr)\geq g(S)$ for all $S\sse V$, and in particular for the set $T$; so
$f(T)\geq g(T)$, which implies that $f(T)=g(T)$. 
\end{proofof}

\section{Tightness of Lemma~\ref{largereq}} \label{largereq-tight}

\begin{figure}[ht!]
\centering
\includegraphics[width=3.75in]{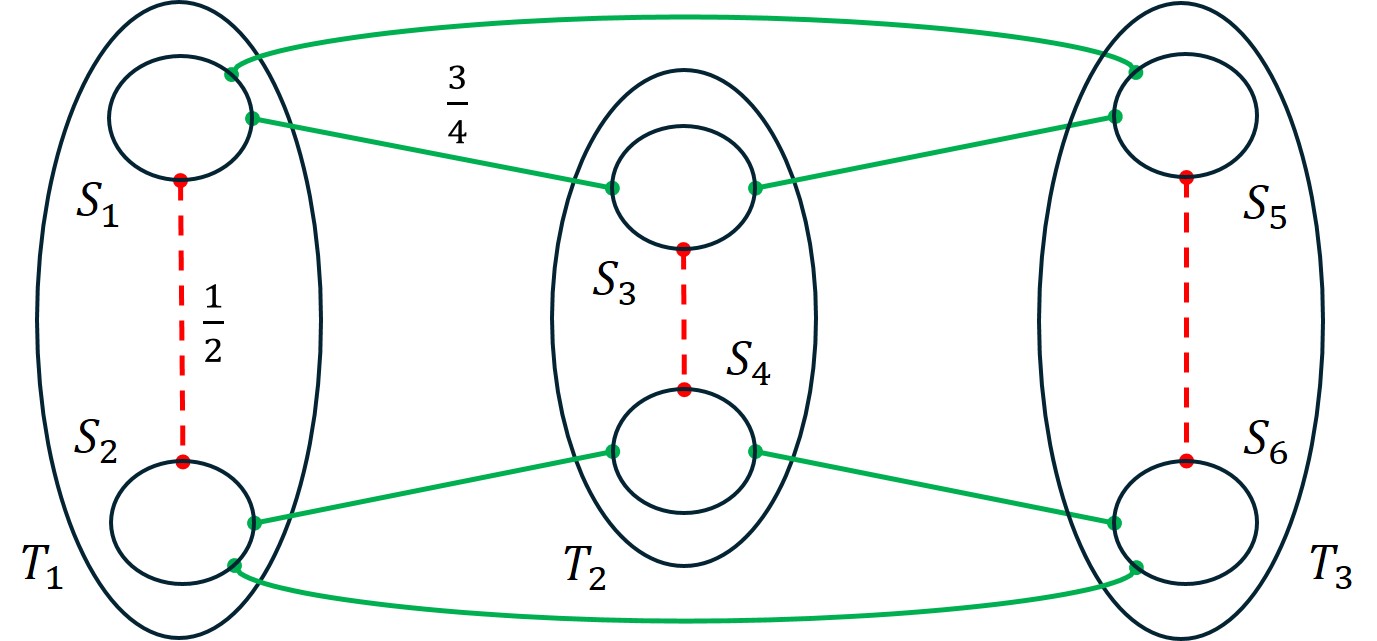}
\captionsetup{font=small, labelfont=small}
\caption{The blobs denote the sets in the laminar family $\Lc$, and the dashed and solid
lines represent the edges in $F$. The dashed (red) edges have $z_e=0.5$, and the solid
(green) edges have $z_e=0.75$. Note that $z\bigl(\dt(S)\bigr)$ is an integer for all
$S\in\Lc$.} 
\label{counterexample_fig}\label{ex1-fig}
\end{figure}

Figure~\ref{counterexample_fig} gives an example showing that the bounds in
Lemma~\ref{largereq} are tight.  
It depicts a laminar family $\Lc$ with $9$ sets, an edge-set $F$ with $9$ edges, and
$z\in(0,1)^F$ such that $b_S:=z\bigl(\dt(S)\bigr)$ is an integer for
all $S\in\Lc$. 
Clearly, we have $|\dt_F(S)|\geq 3$ and $b_S\geq 2$ for all $S\in\Lc$, so to satisfy the
requirements of Lemma~\ref{largereq}, we only need to show that the vectors 
$\bigl\{\chi^{\dt_F(S)}\bigr\}_{S\in\Lc}$ are linearly
independent. We do so by arguing that the solution $z$ shown in the figure is the unique
solution to the system: $x\bigl(\dt(S)\bigr)=b_S\ \ \forall S\in\Lc$.
Since $|F|=|\Lc|$, this shows the desired linear-independence property.

We first argue that all dashed edges must have $x_e=0.5$. 
Consider, for example, the dashed edge $e$ connecting $S_1$ and $S_2$. 
We have  
\begin{equation}
x_e=\frac{b_{S_1} + b_{S_2} - b_{T_1}}{2} = \dfrac{2+2-3}{2}=0.5. \label{ex1ineq}
\end{equation}  
A similar argument shows that all dashed edges have $x_e=0.5$. 
The solid edges can be partitioned into two cycles of length $3$, where the total value of
edges incident to each $S_i$-set is $1.5$. In this case, the only possible solution is a
value of $0.75$ for all solid edges.   

\medskip
We remark that the instance shown in Fig.~\ref{ex1-fig} can arise from an optimal
extreme-point \kecsslp solution after dropping integral edges. 
To show this, we show that the instance can be augmented
by adding edges of $x_e$-value $1$ to obtain an optimal (extreme-point)
\kecsslp solution. 
The simplest way of seeing this is when $k=3$ and each $S_i$ set is a
singleton: if we add an edge with $x_e=1$ parallel to each dashed edge, then one can
verify that the resulting solution is a feasible solution to \kecsslp[3]. 
If we set the cost of the integral edges to $0$, the cost of the dashed edges to $1$, and
the cost of the solid edges to $2$, then this solution is also the unique optimal solution
to \kecsslp[3].%
\footnote{To see this, let $x$ be a feasible solution to \kecsslp[3]. We must have
$x_e=1$ for all the zero-cost edges. Let $D$ be the dashed edges, and $E'$ be the solid
  edges. Since $x\bigl(\dt_{D\cup E'}(S_i)\bigr)\geq 2$ 
for every singleton $S_i$-set, we have $x(D)+x(E')\geq 6$. Similarly, since 
$x\bigl(\dt_{D\cup E'}(T_i)\bigr)\geq 3$, for every $T_i$-set, we have $x(E')\geq 4.5$. So 
$c^Tx\geq 10.5$, and for equality to hold, all these inequalities must be tight; in
particular, all the sets shown in Fig.~\ref{ex1-fig} must be tight.} 
(We can avoid parallel edges by letting each 
$T_i-(S_{i_1}\cup S_{i_2})$ ``ring'' be a clique of size $4$, with one edge each from the
clique to $S_{i_1}$ and $S_{i_2}$. All of these edges have cost $0$ and $x_e=1$.)

\medskip
The instance in Fig.~\ref{ex1-fig} can also result from an optimal extreme-point
\kecsslp solution for {\em even} $k$, after dropping integral edges. Fig.~\ref{ex3-fig}
shows how this may happen for \kecsslp[6], where the thick (blue) edges denote edges with
$x_e=1$. 

\begin{figure}[ht!]
\centering
\includegraphics[width=3.6in]{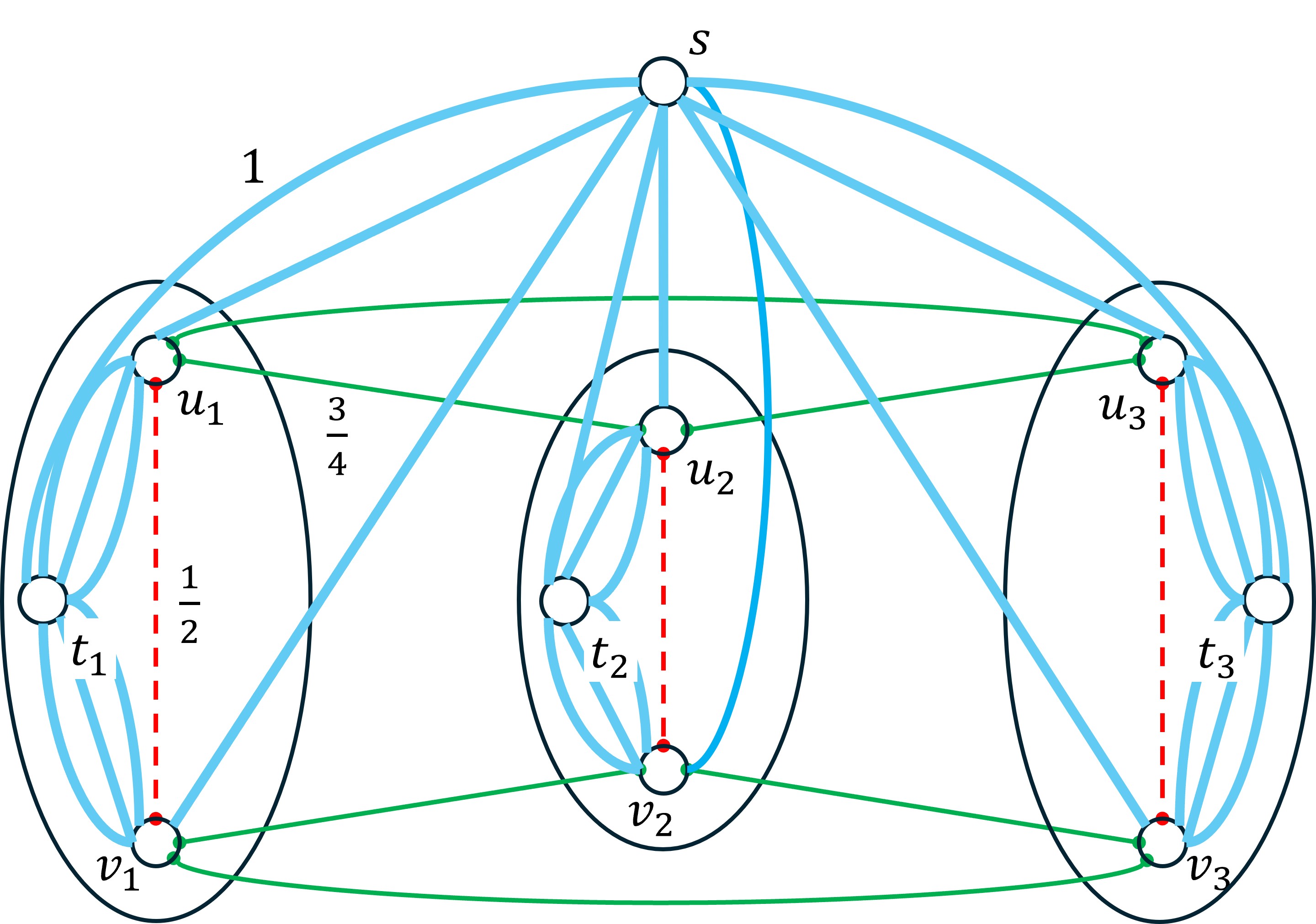}
\captionsetup{font=small, labelfont=small}
\caption{An extreme-point solution to \kecsslp[6]. The dashed (red) edges have
$\bx_e=0.5$, the solid (green) edges have $\bx_e=0.75$, and the thick (blue) edges have
$\bx_e=1$. After dropping integral edges, we obtain the residual instance shown in
Fig.~\ref{ex1-fig}.}  
\label{ex3-fig}
\end{figure}

We briefly justify why the solution $\bx$ shown in Fig.~\ref{ex3-fig} is an extreme-point
solution to \kecsslp[6]. To argue feasibility, it suffices to show that one can send $6$
units of flow 
from $s$ to all other nodes under the $\{\bx_e\}$ edge capacities. By symmetry, we only need
to argue this for the pair $s,t_1$, and the pair $s,u_1$. 
\begin{enumerate}[label=$\bullet$, topsep=0.2ex, itemsep=0.1ex, leftmargin=*]
\item $s,t_1$: We can send $1$ unit of flow along $s,t_1$, $s,u_1,t_1$, $s,v_1.t_1$, and
$0.75$ units of flow along each of $s,u_i,u_1,t_1$ and $s,v_i,v_1,t_1$, for $i=2,3$.

\item $s,u_1$: We can send $1$ unit of flow along $s,u_1$, $s,v_1,t_1,u_1$; 
$0.75$ units of flow along each of $s,u_i,u_1$ and $s,v_i,v_1,t_1,u_1$, for $i=2,3$;
and finally, $0.5$ units of flow along $s,t_1,u_1$ and $s,t_1,v_1,u_1$.
\end{enumerate}

Given feasibility, the fact that $\bx$ is an extreme point follows, because after dropping
integral edges, we obtain the instance in Fig.~\ref{ex1-fig}, with residual
requirements of $2$ for all $\{u_i\}$, $\{v_i\}$ sets, and $3$ for all $\{u_i,v_i,t_i\}$
sets. We argued earlier that the fractional solution shown in Fig.~\ref{ex1-fig} is an
extreme-point solution to the system with these residual requirements. 
Finally, as before, if we set the cost of integral edges to
$0$, the cost of dashed edges to $1$, and the cost of solid edges to $2$, we also obtain
that $\bx$ is the unique optimal solution to \kecsslp[6].

\end{document}